\documentclass[nonacm, sigconf]{acmart}
\makeatletter
\def\@ACM@checkaffil{
	\if@ACM@instpresent\else
	\ClassWarningNoLine{\@classname}{No institution present for an affiliation}%
	\fi
	\if@ACM@citypresent\else
	\ClassWarningNoLine{\@classname}{No city present for an affiliation}%
	\fi
	\if@ACM@countrypresent\else
	\ClassWarningNoLine{\@classname}{No country present for an affiliation}%
	\fi
}
\makeatother
\usepackage{subfigure, paralist, anyfontsize}
\usepackage[table]{colortbl}
\usepackage{pifont}
\usepackage{algpseudocode}
\usepackage[T1]{fontenc}
\usepackage[normalem]{ulem}
\usepackage{enumitem}
\setlist[enumerate]{leftmargin=*}
\setlist[itemize]{leftmargin=*}

\newcommand{\ouralg}{{\texttt{WormHole}}}
\newcommand{\bibfs}{{\texttt{BiBFS}}}
\newcommand{\PLL}{\texttt{PLL}}
\newcommand{\MLL}{\texttt{MLL}}
\newcommand{\SP}{{\texttt{SP}}}
\newcommand{\cmark}{\ding{51}}
\newcommand{\xmark}{\ding{55}}%
	\usepackage[subtle]{savetrees}
    \usepackage{geometry}
\usepackage{multirow}
\usepackage{graphicx}
\usepackage{paralist}
\usepackage{bm}
\usepackage{bbm}
\usepackage{xspace}
\usepackage{url}
\usepackage{fullpage, prettyref}
\usepackage{boxedminipage}
\usepackage{wrapfig}
\usepackage{ifthen}
\usepackage{color}
\usepackage{xcolor}
\usepackage{framed}
\usepackage{algorithm}
\usepackage[normalem]{ulem}
\usepackage[nameinlink]{cleveref}
\usepackage{esvect}
\usepackage{mdframed}
\PassOptionsToPackage{export}{adjustbox}
\usepackage[export]{adjustbox}
\usepackage{thmtools}
\usepackage{thm-restate}

\newtheorem{theorem}{Theorem}[section]

\newtheorem{claim}[theorem]{Claim}

\newenvironment{proofsketch}{%
  \renewcommand{\proofname}{Proof Sketch}\proof}{\endproof}

\newcommand{\ignore}[1]{}

\newcommand{\EX}{\hbox{E}}

\newcommand{\eqdef}{:=}

\newcommand{\davg}{d}

\newcommand{\Sec}[1]{\hyperref[sec:#1]{\S\ref*{sec:#1}}} 
\newcommand{\Eqn}[1]{\hyperref[eq:#1]{(\ref*{eq:#1})}} 
\newcommand{\Fig}[1]{\hyperref[fig:#1]{Figure\,\ref*{fig:#1}}} 
\newcommand{\Tab}[1]{\hyperref[tab:#1]{Table\,\ref*{tab:#1}}} 
\newcommand{\Thm}[1]{\hyperref[thm:#1]{Theorem\,\ref*{thm:#1}}} 
\newcommand{\Fact}[1]{\hyperref[fact:#1]{Fact\,\ref*{fact:#1}}} 
\newcommand{\Lem}[1]{\hyperref[lem:#1]{Lemma\,\ref*{lem:#1}}} 
\newcommand{\Prop}[1]{\hyperref[prop:#1]{Proposition~\ref*{prop:#1}}} 
\newcommand{\Cor}[1]{\hyperref[cor:#1]{Corollary~\ref*{cor:#1}}} 
\newcommand{\Conj}[1]{\hyperref[conj:#1]{Conjecture~\ref*{conj:#1}}} 
\newcommand{\Def}[1]{\hyperref[def:#1]{Definition~\ref*{def:#1}}} 
\newcommand{\Alg}[1]{\hyperref[alg:#1]{Algorithm~\ref*{alg:#1}}} 
\newcommand{\Ex}[1]{\hyperref[ex:#1]{Example~\ref*{ex:#1}}} 
\newcommand{\Clm}[1]{\hyperref[clm:#1]{Claim~\ref*{clm:#1}}} 
\newcommand{\Step}[1]{\hyperref[step:#1]{Step~\ref*{step:#1}}} 

\newcommand{\lz}{\mathcal{C}_{{\textsf{in}}}} 
\newcommand{\lo}{\mathcal{C}_{{\textsf{out}}}} 
\newcommand{\lt}{\mathcal{P}} 
\newcommand{\clc}{\mathcal{C_{\textsf{CL}}}} 
\newcommand{\coregen}{\texttt{CoreGen}} 

\usepackage{tikz}
 \usepackage[usenames,dvipsnames]{pstricks}
 \usepackage{pstricks-add}
 \usepackage{epsfig}
 \usepackage{pst-grad} 
 \usepackage{pst-plot} 
 \usepackage[space]{grffile} 
 \usepackage{etoolbox} 
 \makeatletter 
 \patchcmd\Gread@eps{\@inputcheck#1}{\@inputcheck"#1"\relax}{}{}
 \makeatother
 
 \usepackage{graphicx,amsfonts,amsmath,epsfig,color,multicol,pstricks,tikz, pgf}
\usetikzlibrary{decorations.markings}
\usetikzlibrary{patterns}
\usetikzlibrary{calc, intersections}
\usepackage{pgf,tikz}
\usetikzlibrary{calc}
\usepackage{tkz-euclide}

\definecolor{RoyalPurple}{HTML}{7851a9}
\definecolor{darkred}{HTML}{8B0000}
\definecolor{darkgreen}{HTML}{228B22}
\definecolor{teal}{HTML}{008080}
\definecolor{indigo}{HTML}{4B0082}
\definecolor{lightgray}{HTML}{D7DDDC}
\definecolor{niceblue}{rgb}{0.17, 0.3, 0.6}

 \newcommand{\drawGenerateLzero}[3]{

\node[client, fill=#1] (A) at (0,0) {};
\node[client, fill=#2] (B) at (1,0) {};
\node[client, fill=#2] (C) at (-0.4,0.8) {};
\node[client, fill=#2] (D) at (-0.7,-0.6) {};
\node[client, fill=#2] (P) at (1.4,0.8) {};
\node[client, fill=#3] (E) at (1.7, 1.9) {};
\node[client, fill=#3] (F) at (0.3, 1.6) {};
\node[client, fill=#3] (G) at (-1.8,-1.9) {};
\node[client, fill=#3] (H) at (-1.2,0.9) {};
\node[client, fill=#3] (I) at (1.7,-1.7) {};
\node[client, fill=#3] (J) at (-1.8,1.7) {};
\node[client, fill=#3] (K) at (-0.9,1.9) {};
\node[client, fill=#3] (L) at (-1.9,0.4) {};
\node[client, fill=#3] (M) at (-0.5,-1.9) {};
\node[client, fill=#3] (N) at (0,-1.2) {};
\node[client, fill=#3] (O) at (0.4,-1.6) {};
\node[client, fill=#3] (Q) at (1.2,-0.9) {};
\node[client, fill=#3] (R) at (-1.6,-0.5) {};

\tikzstyle{edgeStyle}=[#2, very thick]
\draw[edgeStyle] (A) -- (B);
\draw[edgeStyle] (A) -- (C);
\draw[edgeStyle] (A) -- (D);
\draw[edgeStyle] (A) -- (P);

\tikzstyle{edgeStyle}=[#3, very thick]
\draw[edgeStyle] (B) -- (P);
\draw[edgeStyle] (B) -- (Q);
\draw[edgeStyle] (C) -- (D);
\draw[edgeStyle] (C) -- (H);
\draw[edgeStyle] (C) -- (F);
\draw[edgeStyle] (C) -- (K);
\draw[edgeStyle] (C) -- (R);
\draw[edgeStyle] (D) -- (R);
\draw[edgeStyle] (D) -- (N);
\draw[edgeStyle] (E) -- (F);
\draw[edgeStyle] (E) -- (P);
\draw[edgeStyle] (G) -- (M);
\draw[edgeStyle] (H) -- (J);
\draw[edgeStyle] (H) -- (K);
\draw[edgeStyle] (H) -- (L);
\draw[edgeStyle] (H) -- (R);
\draw[edgeStyle] (I) -- (Q);
\draw[edgeStyle] (M) -- (N);
\draw[edgeStyle] (N) -- (O);
\draw[edgeStyle] (O) -- (Q);

}

 \newcommand{\drawGenerateLzeroTwo}[3]{

\tikzstyle{client}=[draw=none,circle,minimum size=2.5ex, inner sep=0, fill,]
\def\Figonelen{0.14\textwidth}

\begin{figure*}

      \begin{tikzpicture}
    \begin{scope}[xscale=0.58, yscale=.5, transform shape]
    \drawGenerateLzero{indigo}{indigo!45}{indigo!16};
    \node[scale=1.4, below left] at (B) {1};
    \node[scale=1.4, above right] at (C) {1};
    \node[scale=1.4, below left] at (D) {1};
    \node[scale=1.4, above left] at (P) {1};
    \end{scope}
    
    \begin{scope}[xshift=4cm , xscale=0.58, yscale=.5, transform shape]
    \drawGenerateLzero{indigo}{indigo!45}{indigo!16};
    \node[client, fill=indigo] at (D) {};
    \node[client, fill=indigo!45] at (R) {};
    \node[client, fill=indigo!45] at (C) {};
    \node[client, fill=indigo!45] at (N) {};
    \tikzstyle{edgeStyle}=[indigo, very thick]
    \draw[edgeStyle] (A) -- (D);
    \tikzstyle{edgeStyle}=[indigo!45, very thick]
    \draw[edgeStyle] (D) -- (N);
    \draw[edgeStyle] (D) -- (R);
    \draw[edgeStyle] (D) -- (C);
    \node[scale=1.4, below left] at (B) {1};
    \node[scale=1.4, above right] at (C) {2};
    \node[scale=1.4, above left] at (P) {1};
    \node[scale=1.4, above right] at (N) {1};
    \node[scale=1.4, below left] at (R) {1};
    \end{scope}

    \begin{scope}[xshift=8cm, xscale=0.58, yscale=.5, transform shape]
    \drawGenerateLzero{indigo}{indigo!45}{indigo!16};
    \node[client, fill=indigo] at (D) {};
    \node[client, fill=indigo] at (C) {};
    \node[client, fill=indigo!45] at (R) {};
    \node[client, fill=indigo!45] at (N) {};
    \node[client, fill=indigo!45] at (H) {};
    \node[client, fill=indigo!45] at (F) {};
    \node[client, fill=indigo!45] at (K) {};
    \tikzstyle{edgeStyle}=[indigo, very thick]
    \draw[edgeStyle] (A) -- (D);
    \draw[edgeStyle] (C) -- (D);
    \draw[edgeStyle] (A) -- (C);
    \tikzstyle{edgeStyle}=[indigo!45, very thick]
    \draw[edgeStyle] (D) -- (N);
    \draw[edgeStyle] (D) -- (R);
    \draw[edgeStyle] (C) -- (H);
    \draw[edgeStyle] (C) -- (F);
    \draw[edgeStyle] (C) -- (K);
    \draw[edgeStyle] (C) -- (R);
    \node[scale=1.4, below left] at (R) {2};
    \node[scale=1.4, above left] at (K) {1};
    \node[scale=1.4, above right] at (F) {1};
    \node[scale=1.4, below left] at (H) {1};
    \node[scale=1.4, above left] at (P) {1};
    \node[scale=1.4, below left] at (B) {1};
    \node[scale=1.4, above right] at (N) {1};

    \end{scope}
    \end{tikzpicture}

  \caption{
  Construction of the  inner and outer rings of the core; figure taken from ~\cite{BEOF22}. At any point, the outer ring is the set of vertices adjacent to the inner ring. 
  The algorithm expands the inner ring by adding to it a vertex from the outer ring that has the most neighbors in the inner ring. The image looks at two successive steps: the numbers labelling vertices in the outer ring refer to how many inner ring vertices it is adjacent to. Thus, in the second step, the vertex labelled 2 is added to the inner ring. }
\label{fig:gen-L0}
\end{figure*}
}

\begin{document}
	\title[]{A Sublinear Algorithm 
		for Approximate \\Shortest Paths in Large Networks}
	
	\author{Sabyasachi Basu$^*$}
	\affiliation{%
		\institution{UC Santa Cruz}}
	\email{sbasu3@ucsc.edu}
	
	\author{Nadia Kōshima$^*$}
	\affiliation{\institution{MIT}}
	\email{nadianw36@gmail.com}
	
	\author{Talya Eden}
	\affiliation{%
		\institution{Bar-Ilan University}
	}
	\email{talyaa01@gmail.com}
	
	\author{Omri Ben-Eliezer}
	\affiliation{\institution{MIT}
	}
	\email{omrib@mit.edu}
	
	\author{C. Seshadhri}
	\affiliation{%
		\institution{UC Santa Cruz}
	}
	\email{sesh@ucsc.edu}
	
	\begin{abstract}
		
		Computing distances and finding shortest paths in massive real-world networks is a fundamental algorithmic task in network analysis. There are two main approaches to solving this task. On one hand are traversal-based algorithms like bidirectional breadth-first search (BiBFS), which have no preprocessing step but are slow on individual distance inquiries. On the other hand are indexing-based approaches, which create and maintain a large index. This allows for answering individual inquiries very fast; however, index creation is prohibitively expensive even for moderately large networks. For a graph with 30 million edges, the index created by the state-of-the-art is about 40 gigabytes. We seek to bridge these two extremes: quickly answer distance inquiries without the need for costly preprocessing.
		
		In this work, we propose a new algorithm and data structure, \ouralg{}, for approximate shortest path computations. \ouralg{} leverages structural properties of social networks to build a sublinearly sized index, drawing upon the explicit core-periphery decomposition of Ben-Eliezer et al. [WSDM'22]. Empirically, the preprocessing time of \ouralg{} improves upon index-based solutions by orders of magnitude: for a graph with over a billion edges, indexing takes only a few minutes. Furthermore, individual inquiries are consistently much faster than in BiBFS. The acceleration comes at the cost of a minor accuracy trade-off. Nonetheless, our empirical evidence demonstrates that \ouralg{} accurately answers essentially all inquiries within a maximum additive error of 2 (and an average additive error usually much less than 1). We complement these empirical results with provable theoretical guarantees, showing that \ouralg{}, utilizing a sublinear index, requires $n^{o(1)}$ node queries per distance inquiry in random power-law networks. In contrast, any approach without a preprocessing step (including BiBFS) requires $n^{\Omega(1)}$ queries for the same task.
		
		\ouralg{} offers several additional advantages over existing methods: (i) it does not require reading the whole graph and can thus be used in settings where access to the graph is rate-limited; (ii) unlike the vast majority of index-based algorithms, it returns paths, not just distances; and (iii) for faster inquiry times, it can be combined effectively with other index-based solutions, by running them only on the sublinear core.

	\end{abstract}
	
	\maketitle
	\pagestyle{plain}
	
	\section{Introduction}\label{sec:intro}
	\begin{figure*}[htb!]
		\centering
		\includegraphics[width=.8\textwidth, valign=t]{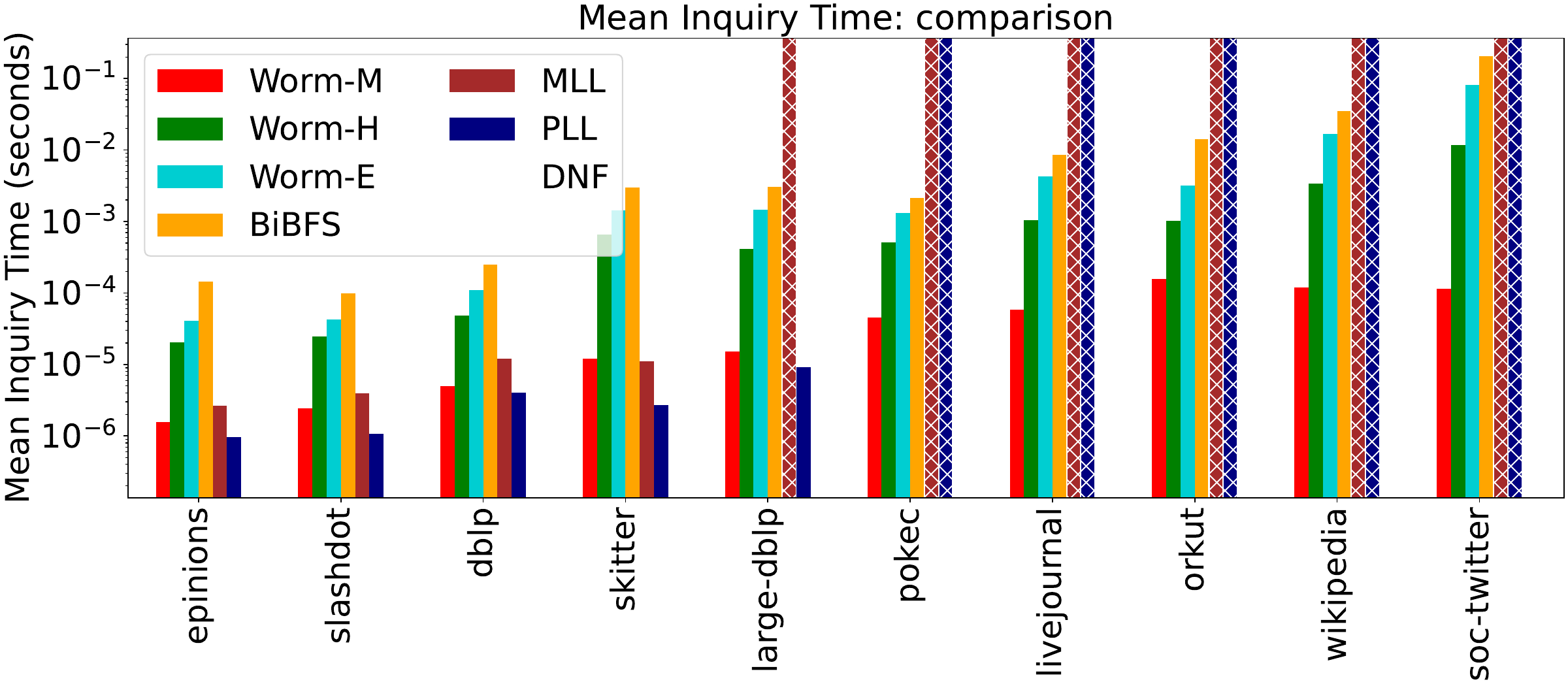}
		
	\caption{We illustrate the average running time per shortest path inquiry for three variants of \ouralg{}, as compared to index-based (MLL \cite{MLL} and PLL \cite{PLL}), and traversal-based (\bibfs{}) competitors. \PLL{} only finds distances, not paths. DNF marks that the preprocessing (index construction) step did not finish. All three of our variants outperformed \bibfs{} consistently. Index based solutions, on the other hand, generally failed on medium to large graphs as the index construction phase timed out. 
		We note that even in smaller graphs where the index construction of MLL and PLL completed successfully, 
		our fastest variant $\ouralg{}_M$ has comparable per-inquiry running time.	}
	\label{fig:exp-summary}
\end{figure*}
Scalable computation of distances and shortest paths in a large network is one of the most fundamental algorithmic challenges in graph mining and graph learning tasks, with applications across science and engineering. Examples of such applications include the identification of important genes or species in biological and ecological networks~\cite{ecology}, driving directions in road networks~\cite{ADGW11, ADFGW11, ADFGW16}, redistribution of task processing from mobile devices to cloud ~\cite{cloud}, 
computer network design and security~\cite{fitch1993shortest,xu2004fighting,4770174, JKSSBKK23},
and identifying a set of users with the maximum influence in a social network ~\cite{KimuraSaito2006,wang2012scalable},  among many others. 
Thus, a long line of work ~\cite{Dijk59, FT87, GH04, PLL, MLL} has developed over the years, constructing   scalable algorithms for distance computation for a variety of real-life tasks.

The simplest methods for answering a shortest path inquiry $(s,t)$ use traversals, among which the most basic is a breadth first search (BFS) starting from $s$ until we reach $t$.  However, the inquiry time for BFS is linear in the network size, which is much too slow for real-world networks.\footnote{To avoid confusion, throughout this paper we use the term \emph{inquiry} to indicate a request (arriving as an input in real-time to our data structure) to compute a short path $\SP(s,t)$ between $s$ and $t$. The term \emph{query} refers to the act, taken by the algorithm itself, of retrieving information about a specific node. For more details on the query model we consider, see~\Sec{setting}.)} 
A  popular modification, Bidirectional BFS ( \bibfs{}), runs BFS from both $s$ and $t$, alternating between the two, until both ends meet. 
It has well been observed in the literature  that \bibfs{} performs surprisingly well for shortest path inquiries on a wide range of networks (see, e.g.,~\cite{MLL,@BM19,bibfs-expander} and the many references within). Because \bibfs{} does not require any prior knowledge on the network structure, it is suitable when the number of shortest path inquiries being made is relatively small. 
However, pure traversal-based approaches  do not scale well when one is required to answer a large number of shortest pair inquiries. As we show in Figure~\ref{fig:exp-summary},  \bibfs{} ends up seeing the whole graph within just a few hundred inquiries.

A long line of modern approaches tackles the distance computation problem in a fundamentally different manner, by preprocessing the network and creating an \emph{index}. The index, in turn, supports extremely fast real-time computation of distances. This line of work has been been investigated extensively in recent years, with Pruned Landmark Labeling (\PLL, Akiba et al.~\cite{PLL}) being perhaps the most influential approach. 

In virtually all index-based methods, pre-processing involves choosing a subset $\mathcal{L}$ of nodes, called \emph{landmarks};  computing all shortest paths among them; and keeping an index of the distance of every node in the network to every landmark. Thus, the space requirement for the index is at least of order $n \cdot |\mathcal{L}|$, where $n$ is the total number of nodes. Naively, this memory requirement can be as bad as \emph{quadratic} in $n$. Despite several improvements to beat the quadratic footprint, existing hub and landmark-based approaches methods continue to have high cost, and can become infeasible even for moderately-sized graphs~\cite{LLMK17}.

Notably, most index-based approaches only return distances in the graph, and not the paths themselves.
The first concrete systematic investigation of solutions outputting shortest paths was made by Zhang et al.~\cite{MLL}. Their work points out that while existing index-based solutions can be adapted to also output shortest paths, these adaptations incur a very high additional space cost on top of that required for distance computations.
The authors of \cite{MLL} then proposed a new approach called monotonic landmark labelling (\MLL) for saving on the index construction space cost. While their algorithm is the current state of the art for this problem, it
is again index-based, meaning that the preprocessing cost is still rather expensive. Improving the computational complexity of the construction phase remains a fundamental challenge.

Beyond the computational constraints, it is sometimes simply unrealistic to assume access to the whole network; examples of scenarios where access is only given via limited query access include, e.g., social network analysis through external APIs \cite{BEOF22}, page downloads in web graphs \cite{CDKLS16}, modern lightweight monitoring solutions in enterprise security \cite{Security2019,Security2022}, and state space exploration in software testing, reinforcement learning and robotics \cite{DART2005}, among many others. 
Existing indexing-based approaches are unsuitable for these scenarios since they require reading the whole graph as a prerequisite. Traversal-based methods such as \bibfs{} are suitable, but as mentioned they do not scale well if one requires multiple distance computations.

The limitations of indexing-based and traversal-based methods give rise to a natural question of whether there is a middle ground solution, with preprocessing that is more efficient than in index-based approaches and inquiry time that is faster than in traversal-based approaches. 
Namely, we ask:
\begin{center}
	\emph{
		Is it possible to answer shortest-path inquiries in\\large networks very quickly, without constructing\\an expensive index, or even seeing the whole graph?}
\end{center}

A general solution for any arbitrary  graph is perhaps impossible; however, real-world social and information networks admit order and structure that can be exploited. In this work we address this question positively for a slightly relaxed version of the shortest path problem on such networks. Inspired by the core-periphery structure of large networks ~\cite{SEIDMAN1983269, RPFM14, Zhang_2015, MGPV20}, we provide a solution which constructs a \emph{sublinearly-sized index} and answers inquiries by querying a strictly sublinear subset of vertices. In particular, our solution \emph{does not need to access the whole network}.  The algorithm returns \emph{approximate} shortest paths, where the approximation error is additive and very small (almost always zero or one). 
In practice, the setup time is negligible (a few minutes for billion-edges graphs), and inquiry times improve on those of \bibfs{}.
Moreover, it can be easily combined with existing index-based solutions, 
to further improve on the inquiry times.
We also include theoretical results that shed light on the empirical performance. 
\begin{figure}[h!]
	\begin{subfigure}[Stored index size on disk for \ouralg{} compared to indexing-based methods.]
		{\includegraphics[width = 0.4\textwidth]{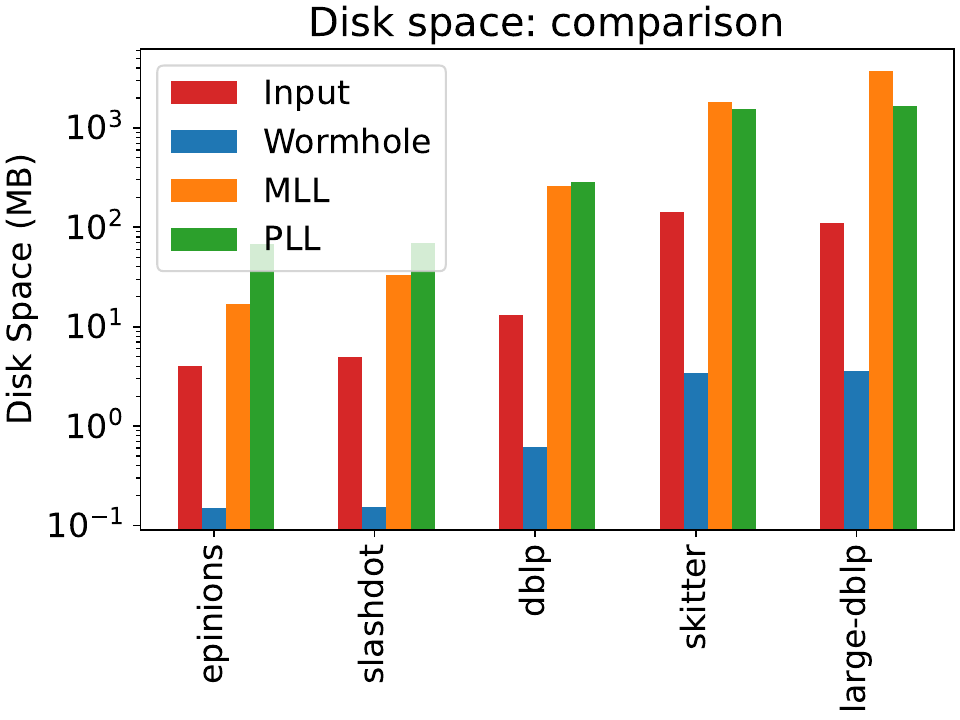}}
	\end{subfigure}
	\begin{subfigure}[Fraction of vertices seen by \ouralg{} in large and huge graphs; see \Tab{classes} for the classes. The dotted lines represent \bibfs{} and solid lines are \ouralg{}.]
		{\includegraphics[width=0.4\textwidth, valign=t]{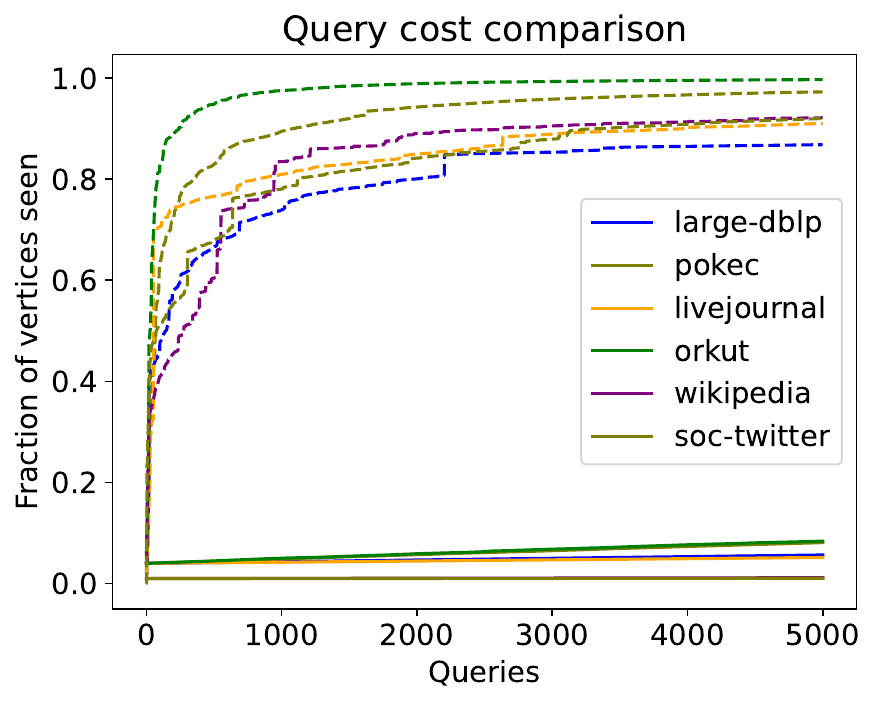}}
	\end{subfigure}
	\caption{(a) a comparison of the footprint in terms of disk space for different methods. The indexing based methods did not terminate on graphs larger than these. For \ouralg{}, we consider the sum of $\lz$ and $\lo$ binary files. Note that \PLL{} here is the \emph{distance} algorithm, solving a weaker problem. The red bar ``Input" is the size of the edge list. 
		(b) we look at the number of vertices queried (visited) by \bibfs{} (dotted lines) and \ouralg{} (solid lines) (the number is the same for all three variants). Observe that while \bibfs{} ends up seeing between 70\% and 100\% of the vertices in just a few hundred inquiries, we are well below 20\% even after 5000 inquiries. 
		\vspace{-0.6cm}}
	\label{fig:summary2}
\end{figure}
\subsection{Our Contribution} \label{sec:contri}
\begin{figure*}[htb!]
	\centering
	
	\begin{subfigure}
		{\includegraphics[width=0.26\textwidth, valign=t]{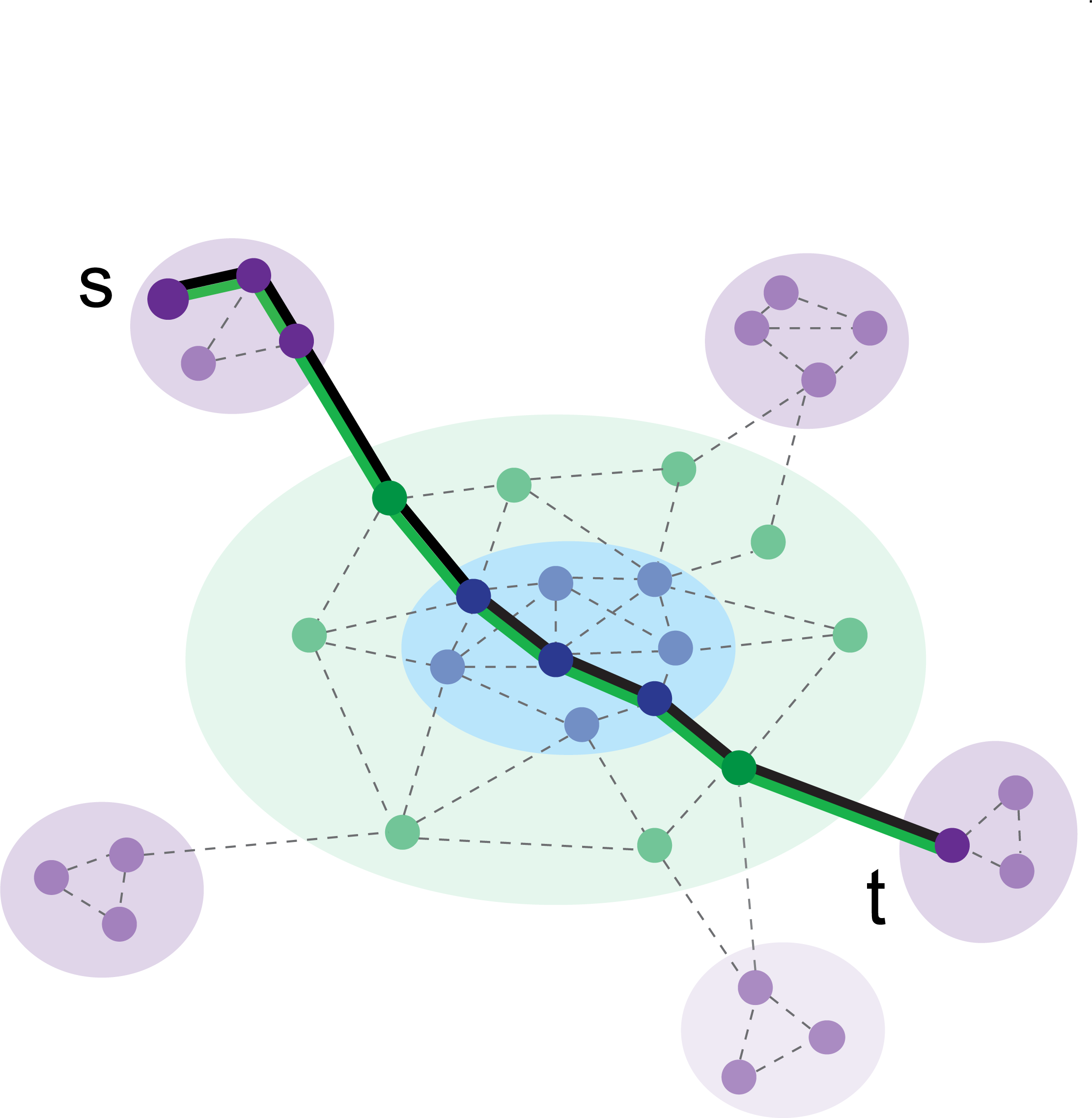}}
	\end{subfigure}
	\hspace{0.05\textwidth}
	\begin{subfigure}
		{\includegraphics[width=0.26\textwidth, valign=t]{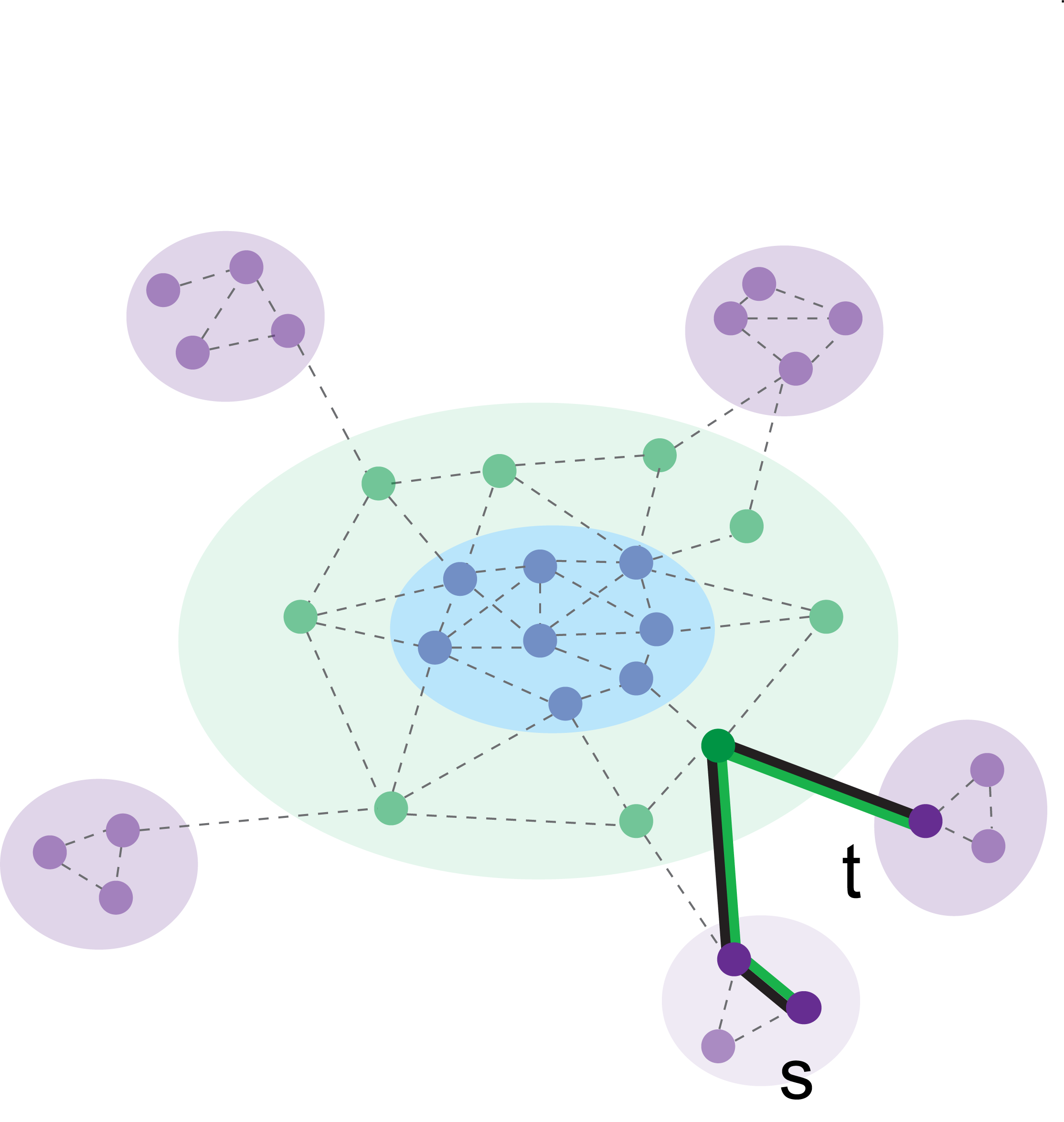}}
	\end{subfigure}
	\hspace{0.05\textwidth}
	\begin{subfigure}
		{\includegraphics[width=0.26\textwidth, valign=t]{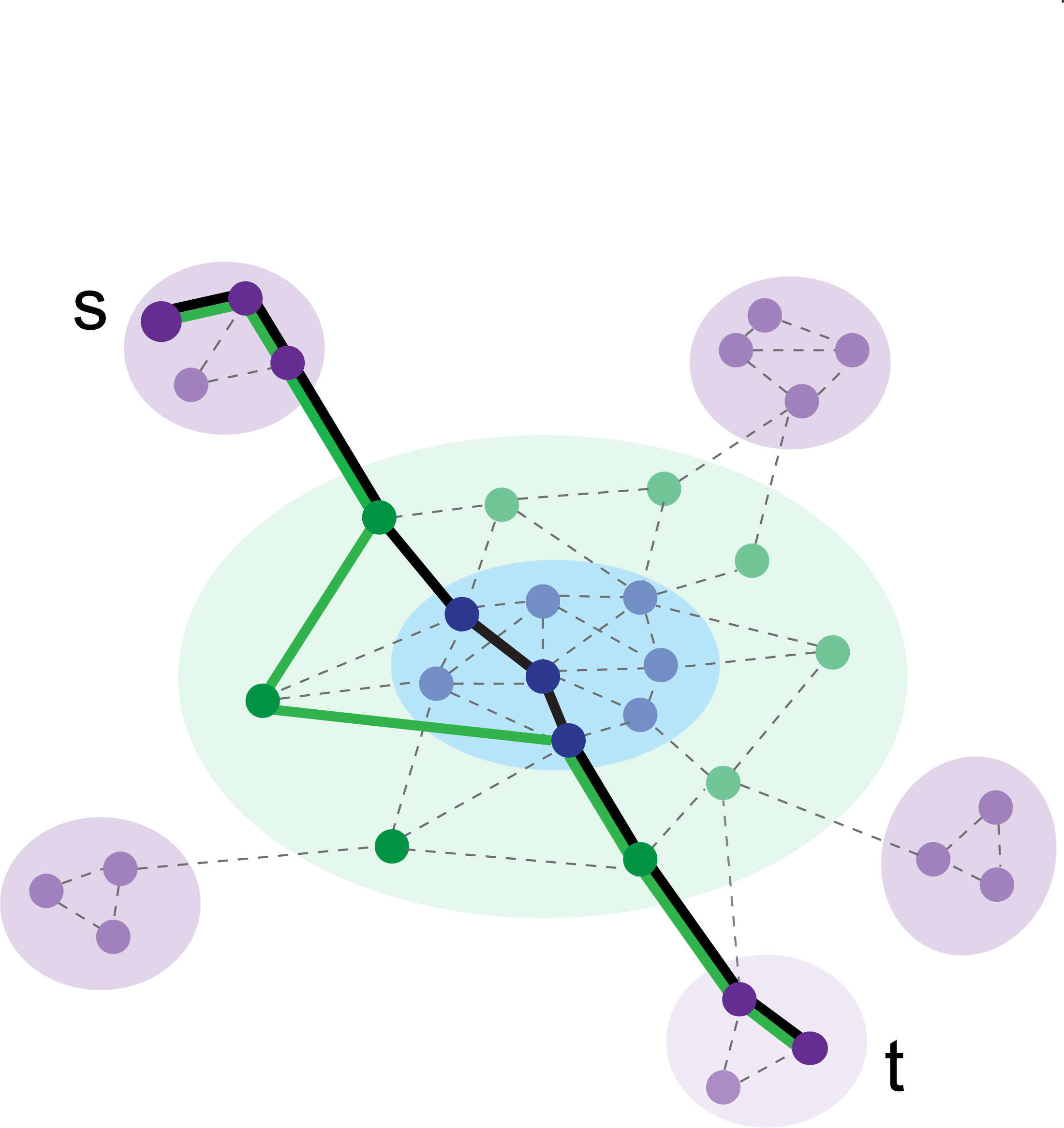}}
	\end{subfigure}
	\caption{The decomposition and some representative cases where \ouralg{} succeeds or fails. The central blue region is the inner ring $\lz$, the green layer outside is the outer ring $\lo$, and purple regions attached to this are the peripheral components forming $\lt$. Dashed lines are edges. Vertices labelled $s$ and $t$ are respectively the source and destination. The green lines are actual shortest paths, while the black lines are paths output by \ouralg. We ignore the case where both source vertices are in the same peripheral component. The first one (a) is the case where the shortest path and the path output by \ouralg{} are identical; no error is incurred in this case. The second (b) is the case where the source and destination are in two different peripheral components, but they encounter a common vertex while traversing to the inner ring. The third (c) is an example of a case where we incur an error: the shortest path(s) interleaves through the outer ring $\lo$, so by restricting the  traversal solely to the interior of $\lz$, we incur an error, in this case of $1$.}
	\label{fig:Algtypes}
\end{figure*}

We design a new algorithm, \ouralg{}, that creates a data structure allowing us to answer multiple shortest path inquiries by exploiting the typical structure of many social and information networks. \ouralg{} is simple, easy to implement, and theoretically backed. We provide several variants of it, each suitable for a different setting, showing excellent empirical results on a  variety of network datasets. 
Below are some of  its key features:
\setlength{\leftmargini}{10pt}
\begin{itemize}
	\item \textbf{Performance-accuracy tradeoff.} 
	To the best of our knowledge, ours is the first approximate \emph{sublinear} shortest paths algorithm in large networks. The fact that we allow small additive error, 
	gives rise to a trade-off between preprocessing time/space and per-inquiry time, and
	allows us to come up with a solution with efficient preprocessing \emph{and} fast per-inquiry time. 
	Notably, our most accurate (but slowest) variant, $\ouralg_{E}$, has near-perfect accuracy: more than $90\%$ of the inquiries are answered \emph{with no additive error}, and in \emph{all} networks, more than $99\%$ of the inquiries are answered with additive error at most $2$. See \Tab{wormhole-summary} for more details.
	\vspace{0.2cm}
	\item \textbf{Extremely rapid setup time.} 
	Our longest index construction time was just two minutes even for billion-edged graphs.
		For context,  \PLL{} and \MLL{} timed out on half of the networks that we tested, and for moderately sized graphs where \PLL{} and \MLL{} did finish their runs, \ouralg{} index construction was $\times100$ faster. Namely, \ouralg{} finished in seconds while \PLL{} took hours.  See \Tab{setupworm} and \Tab{LLsetups}.
	This rapid setup time is achieved due to the use of a sublinearly-sized index. For the largest networks we considered, it is sufficient to take an index of about 1\% of the nodes to get small mean additive error -- see Table \ref{tab:classes}. For smaller networks, it may be up to $6\%$.
	\vspace{0.2cm}
	\item \textbf{Fast inquiry time.} 
	Compared to  \bibfs{}, the vanilla version \ouralg{}$_E$ (without any index-based optimizations) is $\times 2$ faster for almost all graphs and more than $\times4$ faster on the three largest graphs that we tested. A simple variant \ouralg{}$_H$ achieves an order of magnitude improvement at some cost to accuracy: consistently 20$\times$ faster across almost all graphs, and more than 180$\times$ for the largest graph we have.
	See \Tab{wormhole-summary} for a full comparison. 
	Indexing based methods typically answer inquiries in microseconds; both of the aforementioned variants are still in the millisecond regime. 
	\vspace{0.2cm}
	\item \textbf{Combining \ouralg{} and the state of the art.} \ouralg{} works by storing a small subset of vertices on which we compute the exact shortest paths. For arbitrary inquiries, we route our path through this subset, which we call the core. We use this insight to provide a third variant, \ouralg{}$_M$ by implementing the state of the art for shortest paths, \MLL{}, on the core. This achieves inquiry times that are comparable to \MLL{} (with the same accuracy guarantee as \ouralg{}$_H$) at a fraction of the setup cost, and runs for massive graphs where \MLL{} does not terminate. We 
	explore  this combined approach in \Sec{prim}, and provide statistics in \Tab{mlloncore}.
	\vspace{0.2cm}
	\item \textbf{Sublinear query complexity.} The query complexity refers to the number of vertices queried by the algorithm.
	In a limited query access model where querying a node reveals its list of neighbors (see \Sec{setting}), the query complexity of our algorithm scales very well with the number of distance / shortest path inquiries made.
	To answer $5000$ approximate shortest path inquiries, our algorithm only observes between $1\%$ and $20\%$ of the nodes for most networks.
	In comparison, \bibfs{} sees more than $90\%$ of the graph to answer a few hundred shortest path inquiries. 
	See \Fig{summary2} and \Fig{qcost-small} for a comparison.
	\vspace{0.2cm}
	\item \textbf{Provable guarantees on error and performance.} 
	In ~\Sec{theorem} we prove a suite of theoretical results complementing and explaining the empirical performance. The results, stated informally below, are proved for the Chung-Lu model of random graphs with a power-law degree distribution~\cite{chung2002connected,chung2004average,chung2006volume}. 
	\begin{theorem}[Informal]
		In a Chung-Lu random graph $G$ with power-law exponent $\beta\in (2,3)$ on $n$ vertices, \ouralg{} has the following guarantees with high probability:
		\vspace{0.05cm}
		\\ \emph{-- \ Worst case error:} For all pairs $u,v$ of nodes in $G$, the path between $u$ and $v$ output by \ouralg{} is at most $O(\log \log n)$ longer than the shortest path.
		\vspace{0.05cm}
		\\
		\emph{-- \ Query complexity:} \ouralg{} has preprocessing query complexity of $o(n)$, and query cost per inquiry of $n^{o(1)}$.
		\vspace{0.07cm}
		\\
		In contrast, any method that does not preprocess the graph (including \bibfs{}) must have $n^{\Omega(1)}$ query cost per inquiry. 
		\label{thm:informal}
	\end{theorem}
	
\end{itemize}

\subsection{Setting}
\label{sec:setting}
We consider the problem of constructing a data structure for \emph{approximately} answering shortest-path inquiries between pairs of vertices $(s,t)$ in an undirected graph $G$,
given limited query access to the graph. 

\paragraph{Query model.} Access to the network is given through the standard node query model \cite{CDKLS16,BEOF22}, where we start with an arbitrary seed vertex as the ``access point'' to the network,
and querying a node $v$ reveals its list of neighbors $\Gamma(v)$.
Unlike existing index-based solutions, which perform preprocessing on the whole graph, we aim for a solution that queries and stores only a small fraction of the nodes in the network.%

\paragraph{Objective.}
Following the initialization of the data structure, the task is to 
answer multiple shortest path inquiries, where each  inquiry $\SP(s,t)$ needs to be answered with a valid path $p_0 p_1 \ldots p_\ell$ between $s=p_0$ and $t=p_\ell$, and the objective is to minimize the mean additive error measured over all inquiries. 
The additive error for an inquiry $\SP(s,t)$ is the difference between the length of the returned $s$--$t$ path and the actual shortest distance between $s$ and $t$ in $G$.
Depending on the specific application, one would like to minimize (a subset of) the additive error, running time, memory and/or node queries.

\paragraph{Core-periphery structure}
The degree distribution in social and information networks often follows a power-law distribution with exponent $2<\beta<3$, which results in a \emph{core-periphery} structure~\cite{SEIDMAN1983269, Zhang_2015, RPFM14, MGPV20,artico2020rare},  where the core is a highly connected component with good expansion properties, consisting of higher degree nodes, while the periphery is a collection of small, poorly connected components of  low degree.

Our data structure is designed for networks exhibiting these structural characteristics. It takes advantage of the structure by first performing a preprocessing step to acquire (parts of) the core of the network, and then answering approximate shortest path inquiries by routing through the core. The working hypothesis is that pairs of nodes that are sufficiently far apart will typically have the shortest path between them (or close to it) routed through the higher degree parts of the network.
This is somewhat reminiscent of approaches based on the highway dimension ~\cite{ADGW11, ADFGW11, ADFGW16} for routing in road networks, although the structural characteristics of these network types differ considerably.

\subsection{The algorithm}

\ouralg{} builds an explicit hierarchical core-periphery type structure with a sublinear inner ring and provides a framework which uses this structure to answer shortest path inquiries.
There are two phases:
\setlength{\leftmargini}{10pt}
\begin{itemize}
	\item A preprocessing step where we decompose the graph into three partitions, storing only the smallest one: a highly dense subgraph \emph{of sublinear size}.
	\item The phase where we answer inquiries: here the algorithm 
	(approximately) answers shortest path inquiries of the form \SP$(s,t)$ for arbitrary vertex pairs $(s,t)$.
\end{itemize}
We elaborate on the two phases.

\subsubsection{The decomposition}
It is well-documented that social networks exhibit a core-periphery structure; see,  e.g.,~\cite{SEIDMAN1983269, Zhang_2015, RPFM14, MGPV20} and the many references within. The \emph{core} is a  highly-connected component with good expansion properties and smaller effective diameter.
The \emph{periphery}, denoted $\lt$,  consists of smaller isolated communities that connect to the core, but are sparsely connected internally, and whose union is of linear size~\cite{chung2004average}. 
Therefore, when answering shortest path inquiries, it is reasonable  to first check if the two vertices are in the same peripheral community, and otherwise route through the core. 

A recent work of Ben-Eliezer et al.~\cite{BEOF22} observed that a more fine-grained constructive version of the core-periphery decomposition may be useful for algorithmic purposes in real world networks. In their work, the core is further decomposed into two layers: a sublinearly-sized \emph{inner ring}, denoted $\lz$,  which is very dense and consists of the highest degree vertices in the graph; and its set of neighbors, which we refer to as the \emph{outer ring}, $\lo$, where $\lo=\Gamma(\lz)\setminus \lz$. 
Their work shows (empirically) that even a sublinearly-sized inner ring is sufficient so that the union of the inner and outer rings effectively contains the core. Our work makes use of the this fine-grained core-periphery decomposition. During preprocessing, we acquire the nodes of the inner core, which then constitutes the (sublinear) index of our data structure.

\subsubsection{Acquiring the inner core}

To acquire the inner core $\lz$ we follow the procedure by~\cite{BEOF22}. 
The procedure gradually expands $\lz$, starting from an arbitrary seed vertex, 
and  iteratively adding vertices with high connectivity to the current core.
See \Alg{core} for the pseudo code and~\Fig{gen-L0} for an illustration of an expansion step.

\subsubsection{Answering shortest-path queries}

In the second phase, given a query  $\SP(s,t)$, \ouralg{} does the following. First, it checks if the two vertices are in the same peripheral component, by performing a truncated  \bibfs{} from both $s$ and $t$ up to depth two.  If the two trees collide, it returns the shortest path between $s$ and $t$. Otherwise, \ouralg{} continues both BFS traversals until it reaches the outer ring (from both $s$ and $t$). From here, it  takes a single step to reach the inner ring, and then performs a restricted  \bibfs{} on the subgraph induced by the inner ring vertices. We note that the choice of  \bibfs{} here is arbitrary, and we can use any shortest-path algorithm (including modern index-based approaches, initialized only on the inner core) as a black-box to find a shortest path in the inner ring. 

\Fig{Algtypes} illustrates a few typical cases encountered by the algorithm;  in the first two cases the algorithm returns a true shortest path, and in the third case the returned path is not a shortest path (thus incurring a nonzero additive error). 

We stress that a single decomposition is  subsequently used to answer \emph{all} shortest path queries. \Thm{informal} provides a strong theoretical guarantee on the performance of \ouralg{}. It is worth emphasizing that our notion of approximation is inspired by 
practical relaxations, and is
distinct from the one usually considered in theoretical works.

\section{Related Work}
There is a vast body of research on shortest paths and distances, spanning over many decades, and including hundreds of algorithms and heuristics designed for a variety of settings. Here, we only review several works that are most closely related to ours. For a more comprehensive overview, see the surveys~\cite{madkour2017survey,sommer2014shortest} and references therein.

\paragraph{Index-based approaches}
As mentioned earlier, a ubiquitous set of algorithms are based on  landmark/sketch  approaches \cite{ZYG19, li2020scaling, MLL}, with Pruned Landmark Labeling (\PLL) \cite{PLL} being perhaps the most influential one.
These algorithms follow a
two-step procedure: the first step generates an ordering of the vertices according to importance (based on different heuristics), and  the the second step
generates labeling from pruned shortest path trees constructed according to the ordering. 
	Then the shortest distance between an arbitrary pair of
	vertices $s$ and $t$ can be answered quickly based on their labels.
	However, even with pruning,  \PLL{} requires  significant setup time.
	Hence,  there have been many attempts to parallelize it ~\cite{psl, vc-pll, li2020scaling}.

	\paragraph{Embedding based approaches} Some recent approaches leverage embeddings of graphs to estimate shortest paths. Like in representation learning, they seek to find efficient representations of distances between pairs of nodes~\cite{orion, rigel}. A modern line of work also considers hyperbolic embeddings of the graphs, that are closely related to tree decompositions, to answer shortest path inquiries ~\cite{blasius, Kitsak_2023}. Recent work has also looked at accelerating this process by using GPU based deep learning methods~\cite{deep-learning, learningdist2, Kutuzov_2019}.
	Query-by-Sketch~\cite{QBS}
	considers the, related but incomparable, task of 
	answering \emph{shortest-path-graph inquiries}, where the goal  is to compute a subgraph containing exactly all shortest paths between a given pair of vertices.
	They propose an alternative labeling scheme to improve the scalability and inquiry times.

	\paragraph{BFS-based approaches}
	\sloppy
	Another set of algorithms is based on  Breadth First Search (BFS) or Bidirectional Breadth First Search (\bibfs{}), as they are exact methods with no preprocessing step. There has been substantial work in the last few years proving that  \bibfs{} is sublinear (i.e., proving complexity upper bounds of the form $n^{1-\Omega(1)}$) with high probability for several graph families. These include, e.g., hyperbolic random graphs
	(Blasiüs et al., \cite{blasius}),
	and graphs with a finite second moment and power-law graphs (Borassi and Natale~\cite{@BM19}).
	Very recently, Alon et al.~\cite{bibfs-expander} showed that \bibfs{} has sublinear query complexity for a broad family of expander graphs. Their work also proves query lower bounds of the form $n^{\Omega(1)}$ for traversal-based algorithms in Erd\H{o}s-Renyi and $d$-regular random graphs.
	
	\paragraph{Core-periphery based approaches}
	Several other works  exploit the core-periphery structure of networks ~\cite{BC06, ASK12, CTL}. Brady and Cohen~\cite{BC06} use compact routing schemes to design an algorithm with small additive error based on the resemblance of the periphery to a forest.  Akiba, Sommer and Kawarabayashi~\cite{ASK12} exploit the property of low tree-width outside the core, and in ~\cite{CTL}, the authors design a Core-Tree based index in order to improve preprocessing times on massive graphs. We note that in all these results, the memory overhead is super-linear.

	\paragraph{Worst case graphs}
	On the theoretical side, there have been many results on exact and approximate shortest paths in worst-case graphs, e.g., ~\cite{ASP1, ASP2, ASP3, ASP4, ASP6, ASP7},  with improvements as recently as the past year. Most of these focus on the $2$-approximation case, first investigated in the seminal work of Aingworth et al. ~\cite{ASP8}. We point the readers to  Zwick's  survey~\cite{Zwick01} on
	exact and approximate shortest-path algorithms, and 
	Sen's survey ~\cite{sen2009approximating} on distance oracles
	for general graphs with an emphasis on lowering pre-processing cost. 
	Notably, because we make a beyond worst case structural assumption that is common in large networks, namely, a core-periphery structure, both our results and algorithm substantially differ from the worst-case theoretical literature.

	\section{Algorithm}
	\ouralg{} utilizes insights about the structure of real world networks to cleverly decompose the graph and calculate approximate shortest paths. We discuss the algorithm and various steps in detail in this section; our two main components are a sublinear decomposition procedure, adapted from the recent work of Ben-Eliezer et al~\cite{BEOF22}, and a routing algorithm that takes advantage of this decomposition to find highly accurate approximate shortest paths. 
	
	\drawGenerateLzeroTwo{black}{black!48}{black!20}
	
	\subsection{The Structural Decomposition Phase}
	The former is a simple implementation of a structural decomposition provided in ~\cite{BEOF22} that gives us access to the inner ring; the inner ring is central to our procedure, and the bulk of our computation is done on it. This decomposition stratifies our graph into three sections: the inner ring, $\lz$, a dense component with the highest degree nodes, the outer ring ,$\lo$, the set of  neighbors of the inner ring, and the rest of the vertices which form the periphery, $\lt$, and  typically reside in fragments of size $O(\log n)$~\cite{LuThesis}. The construction of the inner ring is an iterative procedure that captures the highest degree vertices in sublinear time; we shall refer to this process as the \coregen{} procedure. The procedure gradually expands the inner ring, starting from an arbitrary seed vertex. The outer ring is the set of neighbors of the inner ring vertices (that are not already in the inner ring). At each step, the procedure removes, from $\lo$, the vertex with highest number of neighbors in $\lz$  and adds it to $\lz$. It then queries this vertex and adds its neighbors to $\lo$, while keeping track of how many neighbors $\lo$ vertices have in the updated $\lz$.  See~\Fig{gen-L0} for the first few steps of the inner ring generation procedure. The size of inner ring, in terms of a percentage of the vertex set, is given to the core-generation algorithm as a parameter and the process is iterated till  that size is reached. \Alg{core} has the pseudocode for \coregen. Here, $\Gamma(v)$ refers to the neighbors of a vertex $v$.
	
	We emphasize that \emph{the decomposition is performed only once}, and all subsequent operations are done using this precomputed decomposition. 
	
	\begin{algorithm}
		\caption{\texttt{CoreGen}$(v,s \;\small{(<1)},\text{ query access to } G=(V,E))$: starting from a seed vertex $v$, generates an inner ring of size $s|V|$; returns an inner ring, $\lz$, and an outer ring, $\lo$.}
		\begin{algorithmic}[1]
			\State size = 0
			\State $\lz\leftarrow v$, $\lo \leftarrow \Gamma(v)$
			\While{size < $s|V|$}
			\State Pick $u\in \lo$ with maximum number of neighbors in $\lz$ (break ties randomly)
			\State $\lz = \lz\cup \{u\}$, $\lo = \lo \cup \Gamma(u)\setminus \lz$
			\State size $+= 1$
			\EndWhile
			\State \Return $\lz, \lo$
		\end{algorithmic}
		\label{alg:core}
	\end{algorithm}
	
	\begin{algorithm}
		\caption{\texttt{\ouralg$(s,t, \lz, \lo, G)$}: final algorithm }
		\begin{algorithmic}[1]
			\State $T(s) =\Gamma(s), T(t) = \Gamma(t)$
			\State $C(s), C(t)\leftarrow \varnothing$
			\While{$T(s)\cap T(t) =\varnothing$ or if both $C(s), C(t) = \varnothing$}
			\For{both $s$ and $t$ (denoted by $i$ in the next block) }
			
			\If{$C(i) = \varnothing$ and $\Gamma(T(i))\cap \lz\neq \varnothing$}
			\State Expand $T(i)$ by a single level
			\State $C(i)\leftarrow \Gamma(T(i))\cap \lz$
			\State \textbf{stop} further expansion of $T(i)$
			\ElsIf{$C(i)=\varnothing$}
			\State Expand $T(i)$ by a single level
			\EndIf{}
			\EndFor{}
			\EndWhile{}
			\State
			\If{$T(s)\cap T(t) \neq\varnothing$}
			\State\Return \SP$(s,t)$ through $T(s)\cap T(t)$
			\Else{}
			\State Expand the \bibfs{} tree from  $C(s), C(t)$ in $G[\lz]$
			\State\Return \SP$(s,t)$ through $T(s)\cup G[\lz]\cup T(t)$
			\EndIf{}
		\end{algorithmic}
		\label{alg:wormhole}
	\end{algorithm}
	
	\subsection{The Routing Phase}
	The approach is simple: 
	assume the preprocessing phase acquires  the inner ring.  Upon an inquiry $(s,t)$, the algorithm  starts two BFS trees from both $s$ and $t$. It then expands the tree at each step till one of the followings happens:
	\begin{compactenum}
		\item the two trees intersect, or
		\item the trees reach the outer ring.
	\end{compactenum}
	
	If the search trees in the former do not intersect, \ouralg{} mandatorily routes shortest paths through the inner ring. Once in the inner ring, it computes the exact shortest path through it. 
	
	\subsection{Variants of \ouralg}
	In the  default  variant of \Alg{wormhole}, \ouralg{}$_E$, we use   the bidirectional breadth-first \bibfs{} shortest path algorithm as a primitive in order to compute the shortest path between two inner ring vertices; see~\cite{Pohl69} for a full description of \bibfs.  The theoretical analysis in \Sec{theorem} considers  this variant.
	
	\ouralg{}$_H$ is a  simple variant of \Alg{wormhole} where we start the expansion of the \bibfs{} trees in $G[\lz]$ only from the vertex $v_s=\text{argmax}_{v\in C(s)}deg(v)$ (and similarly $v_t$ for $C(t)$, the highest degree vertices). 
	We note that for this variant, 
	some of the theoretical results in \Sec{theorem} no longer hold:
	while \ouralg{}$_H$  satisfies all sublinearity properties, the error bounds in \Thm{add-error} no longer hold, as this variant only computes an approximate shortest path inside the inner ring. An empirical comparison of the two is done in \Sec{vsbibfs} and \Tab{wormhole-summary}.
	
	In \ouralg{}$_{M}$ we combine \ouralg{} with 
	an index-based algorithm restricted to the core. 
	While indexing-based algorithms are quite expensive, in terms of both preprocessing and space cost, they lead to much lower times per inquiry. Therefore this variant is suitable when faster inquiry times are preferable to low space requirements.
	\ouralg{}$_{M}$ makes the index creation cost substantially lower  (compared to generating it for the entire graph) while providing speedups for answering shortest path inquiries  compared to the \bibfs{} implementation. 
	We discuss this briefly in \Sec{prim} but leave a complete systematic exploration of these options for future research.

	\section{Theoretical Analysis} \label{sec:theorem}
	
	It is a relatively standard observation that many social networks exhibit, at least approximately, a power-law degree distribution (see, e.g.,~\cite{artico2020rare} and the many references within). The Chung-Lu model~\cite{LuThesis} is a commonly studied random graph model which admits such degree distribution. 
	
	In this section we provide a proof-of-concept for the correctness of \ouralg{} on Chung-Lu random graphs, aiming to explain the good performance in practice through the study of a popular theoretical model. 
	We sometimes only include proof sketches instead of full proofs, in the interest of saving space.

	\subsection{Preliminaries}
	
	We start by defining power-law networks and the Chung-Lu model~\cite{chung2002connected,chung2004average,chung2006volume}, followed by a set of useful results from Lu~\cite{LuThesis}.
	A network is said to have a power-law degree distribution when for sufficiently large $k$, the fraction of nodes with degree $k$, denoted $p(k),$ follows a power-law. That is, $p(k) \propto k^{-\beta}$. Multiple studies have shown that real world graphs indeed follow (or approximately follow) a power-law degree distribution, and the exponent $\beta$ is typically in the range $2<\beta<5$ or so ~\cite{Newman03, FFF99, Albert_1999, SCMRSBC21}. Theoretically, the regime $2 < \beta < 3$ is the most interesting, and thus we focus our study on this regime.
	
	In a series of works, Chung and Lu suggested a model to generate graphs according to a power-law degree distribution, in order to generate massive graphs that capture properties of real-world graphs~\cite{chung2002connected,chung2004average,chung2006volume}. 
	In this model, 
	one is given a list of $n$ expected degrees,  $w_1\geq w_2 \geq \ldots \geq w_n$, where in our case we assume that the expected  degrees follow 
	a  power-law degree distribution with $2<\beta<3$.  Then, for each pair of nodes $i,j$, the edge between them is instantiated with probability $\frac{w_i\cdot w_j}{W}$, where $W=\sum_{i=1}^n w_i$ (and it is assumed that for any $i,j$, $w_i\cdot w_k\leq W$). Note that indeed by this definition, it holds that for every vertex  $v$, $\EX[d(v)]=w(v)$, and  therefore that $\EX[|E(G)|]=W/2$.
	We denote this distribution on graphs as $\mathcal{G}_{CL}$, and  denote a graph drawn from it by $G\sim \mathcal{G}_{CL}$.

	In all that follows we assume that $G\sim\mathcal{G}_{CL}$, and the probabilities are taken over the generation process.
	We let $\gamma=1/\log \log n$, and let $\clc$ denote the set of vertices with degree at least $t=n^\gamma$.
	
	\begin{theorem}[Theorem 4 in~\cite{chung2004average}]\label{thm:chung-lu}
		Suppose a power law random graph with exponent $2<\beta<3$, average  degree  $d$  strictly  greater  than  1, and  maximum  degree  $d_{max}>\log n/\log\log n$. Then almost surely the diameter is $\Theta(\log n)$, the diameter of the $\clc$ core is $O(\log \log n)$ and almost all vertices are within distance $O(\log \log n)$ to $\clc$.
	\end{theorem}

	\begin{claim}[Fact 2 in~\cite{chung2002connected}]\label{clm:exp-deg}
		With probability at least $1-1/n$, for all vertices $v\in V$ such that $w(v)\geq 10\log n$
		\[d(u)\in\left[\frac{1}{2}w(u), 2w(u)\right]
		\text{ and }
		|E(G)|\in [W/4,W]\;,\]
		and for all vertices with $w(v)\leq 10\log n$, $d(v)\leq 20\log n.$
	\end{claim}

	\subsection{Sublinearity of Inner Ring}
	Chung and Lu proved ~\cite{chung2002connected, chung2004average, chung2006volume} that in random graphs with a power-law degree distribution and exponent $2<\beta<3$, there exists a core  with small diameter, so that most of the vertices in the graph are close to it. Here we extend their result to show that this core is of sublinear size and that \coregen{} indeed captures it; i.e., $\clc \subseteq \lz$. 
	
	\begin{theorem}[\ouralg{} inner ring $\supseteq$ Chung-Lu core]\label{thm:clc_in_lz}
		Consider a graph $G$ generated according to the Chung-Lu process with $\beta > 2$, and suppose we run \ouralg{} on $G$ with an inner ring of size $n^{1-\Theta(1/\log\log n)}$. Then with high constant probability, $\lz$ contains all vertices with degree greater than $n^{1/\log \log n}$. 
	\end{theorem}
	\begin{proof}[Proof Sketch]
		
		Consider a process $P$ which simultaneously constructs $G$ as the core grows as outlined in \coregen{}: 
		Initially the graph $G$ has no edges, and $\lz$ is an empty set. Every time a vertex $v$ is added to $\lz$, $P$ iterates over all nodes $u \in V\setminus \lz$ and draws a coin with probability $\frac{w(v)\cdot w(u)}{2W}$ to determine whether the edge $(u,v)$ is in $E$ or not. Once $\lz$ is acquired  (i.e., when $\lz$ is of size $t=100 \log n\cdot n^{1-1/\log\log n}$), $P$ iterates over  all pairs of vertices that were not yet examined and determines their edges.
		Clearly this process generates a graph according to the distribution $G\sim\mathcal{G}_{CL}$ (since we only changed the order in which the edges were decided).  
		
		Let $\davg$ denote the expected average degree in the graph, and let $\lz^{\ell}$ denote the set of vertices added to $\lz$ after $\ell$ steps of the above process. Finally, let $\gamma=1/\log\log n$.
		It is easy to show that with high probability, except for a negligible fraction, all vertices added to the core have expected degree $w(V)\geq \davg$. Conditioned on this event, 
		for every $u$ with $w(u)\geq n^{\gamma}$, 
		every time a vertex $v$ is added to the core, $v$ creates an edge with $u$ with probability 
		\[\frac{d(u)\cdot d(v)}{W}\geq \frac{n^\gamma\cdot \davg}{n\davg}=\frac{1}{n^{1-\gamma}}.
		\]
		Let $\chi_i$ denote the event that the $i$-th vertex added to $\lz$ creates an edge with $u$. Then by the above, $\EX[\chi_i]\geq \frac{1}{n^{1-\gamma}}$.
		Observe that  $d_{\lz}(u) = \sum_{i} \chi_i =\frac{|\lz^{\ell}|}{n^{1-\gamma}}$, 
		and set $\mu \eqdef \frac{|\lz^{\ell}|}{n^{1-\gamma}}$. Note that the $\chi_i$ variables are independent $\{0,1\}$ random variables, and that  for $\ell=c\log n\cdot n^{1-\gamma}$,
		$\mu\geq c\log n$. Therefore, by the multiplicative Chernoff bound, 
		
		\begin{align*}
			\Pr&\left[\left|d_{\lz^{\ell}}(u)-\mu\right|> \frac{1}{4}\mu \right] 
			< 2\exp\left(\frac{-\frac{1}{16}\cdot \mu}{3}\right)\leq \frac{1}{n^2}.
		\end{align*}
		
		That is, we get that with probability $1-1/n$, all vertices with expected degree $w(u)\geq n^{\gamma}$ have $d_{\lz}(u)\geq \frac{3}{4}n^{\gamma}$.
		A similar analysis can be used to prove that all vertices $u$ with expected degree $w(u)< n^{\gamma}/8$ will have $d_{\lz^{2\ell}}(u)\leq \frac{1}{2}n^{\gamma}$ with high probability.
		Hence, after performing additional $\ell$ steps, we will have that with high probability, all vertices with expected degree  at least $n^{\gamma}$ will be added to $\lz^{2\ell}$ before any vertex with expected degree less than $n^{\gamma}/8$. The proof concludes by noticing that by Lemma~\ref{clm:exp-deg}, the degrees of all vertices will either be as expected 
		up to a constant multiplicative factor, or have an expected degree below, say $n^{\gamma} / 8$, and actual degree not higher than $n^{\gamma} / 2$.
	
	\end{proof}
	\subsection{Approximation Error}
	Now that we have a sublinear inner ring that contains the Chung-Lu core, we must show that routing paths through it incurs only a small penalty. Intuitively, the larger the inner ring, the easier this is to satisfy: if the inner ring is the whole graph, the statement holds trivially. Therefore the challenge lies in showing that we can achieve a strong guarantee in terms of accuracy even with a sublinear inner ring. We prove that \ouralg{} incurs an additive error at most $O(\log \log n)$ for all pairs, which is much smaller than the diameter $\Theta(\log n)$. 
	
	\begin{theorem}[Good additive error]\label{thm:add-error}
		If $\clc \subseteq \lz$, then with high probability, for all pairs $(s,t)$ of nodes, the additive error of our algorithm for the inquiry $\SP(s,t)$ is at most 
		$O(\log\log n).$
	\end{theorem}
	
	The above result holds with high probability even \emph{in the worst case}. 
	Namely, for \emph{all} pairs $(s,t)$ of vertices in the graph, the length of the path returned by \ouralg{} is at most $O(\log \log n)$ higher than the actual distance between $s$ and $t$. This trivially implies that the average additive error of \ouralg{} is, with high probability, bounded by the same amount.
	
	\begin{proof}
		Let $d(s,\clc)$ and $d(\clc,t)$ denote the distances of $s$ and $t$ to the core. If $d(s,t) \leq d(s,\clc) + d(\clc,t)$, then the BFS trees from both sides will intersect without the need to go through the core $\clc$.

		For pairs $(s,t)$ where there is a shortest path going through $\clc$, the length of this shortest path is 
		$$
		d(s,u) + d(u,v) + d(v,t)
		$$
		where $u$ and $v$ are in the core, whereas the path that \ouralg{} outputs is of length at most
		$$
		d(s,\clc) + \text{diam}(\clc) + d(\clc, t).
		$$
		By definition, $d(s,\clc) \leq d(s,u)$ and $d(\clc,t) \leq d(v,t)$. The proof follows since the diameter of $\clc$ is bounded by 
		$O(\log\log n)$
		with high probability (see Theorem~\ref{thm:chung-lu}).
	\end{proof}
	
	\subsection{Query Complexity}
	Recall the node query model in this paper (see \Sec{setting}): starting from a single node, we are allowed to iteratively make  queries, where each query retrieves the neighbor list of a node $v$ of our choice. We are interested in the query complexity, i.e., the number of queries required to conduct certain operations.
	
	Using the inner ring as our index, and routing the shortest paths through it, we get an algorithm with small additive error and sublinear setup query cost: 
	we prove that our query cost per inquiry is also subpolynomial (i.e., of the form $n^{o(1)}$, where the $o(1)$ term tends to zero as $n \to \infty$). Moreover, we prove that  preprocessing is necessary to achieve such query complexity per inquiry; anything else \emph{must} incur a cost of $n^{\Omega(1)}$.
	
	The first result is the upper bound on our performance. 
	
	\begin{theorem}[Subpolynomial query complexity for shortest paths]
		Suppose that the preprocessing phase of our algorithm acquires $\lz$.
		The average query complexity of computing a path between a pair of vertices is bounded by $n^{\Theta(1 / \log \log n)}$.
	\end{theorem}
	\begin{proofsketch}
		For a given inquiry $\SP(u,v)$, we give an upper bound on the query complexity of the BFS that starts at $u$, and similarly for $v$; the total query complexity is the sum of these two quantities. 
		
		Let $C$ denote the subset of vertices obtained by the algorithm during preprocessing. If $u \in \lz$, then the algorithm only performs a \bibfs{} restricted on $G[\lz]$ to another vertex $v\in \lz$.
		Since $\lz$ is queried during the preprocessing phase, this step requires no additional queries.
		The non-trivial case is when the source $u$ is not in $\lz$.
		In this case, the algorithm performs a BFS until it either collides with the simultaneous  BFS that is taking place from the other vertex in the inquiry, or until it reaches     $\lz$, at which point it performs  a \bibfs{} inside $\lz$. 
		
		Fix some vertex $u \notin \lz$. 
		By the definition of the Chung-Lu model, when performing a walk in the graph from some vertex $v$, for \emph{every  step in the walk}, the probability it reaches $\lz$ is at least
		\begin{align}
			p& = \sum_{w\in \lz}  \frac{d(w)}{2m} =\frac{Vol(\lz)}{2m}
			\geq \left(\frac{|\lz|\cdot d}{2n\cdot d}\right) = \frac{1}{n^{\Theta(1/\log\log n)}}.
		\end{align}
		Therefore,  after expanding the BFS tree from $u$ for $O(1/p)$ many times, we reach $\lz$ with high constant probability. Hence, the BFS from $u$ requires $\Theta(1/p)=n^{\Theta(1/\log\log n})$ many queries until reaching $\lz$. This concludes the proof in the case that all shortest paths between $u$ and $v$ contain at least one edge in $\lz$.
		
		It remains to consider the case in our algorithm where the searches from $u$ and $v$ collide outside $\lz$, or reach the same vertex in $\lz$. This requires us, at most, to query all vertices encountered during the BFS until it reaches $\lz$ for the first time. By Theorem~\ref{thm:clc_in_lz}, the degree of all vertices outside $\lz$ is bounded by $O(n^{1/\log \log n})$. Therefore the overall query complexity of the walk outside the core is $ n^{1/\log \log n}\cdot  n^{\Theta(1/\log\log n)}=n^{\Theta(1/\log\log n)}$.
	\end{proofsketch}
	
	Finally, this brings us to the lower bound. We prove that any method for finding a path between two nodes $u$ and $v$ in a Chung-Lu random graph that does not employ preprocessing, requires $n^{\Omega(1)}$ node queries to succeed with good probability.
	Due to space considerations we only provide here a proof sketch; for a more complete proof of similar results for Erd\H{o}s-Renyi and other random graphs, see Alon et al.~\cite{bibfs-expander}.
	
	\begin{theorem}[polynomial query complexity without preprocessing]
		Fix $2 < \beta < 3$. Let $G$ be a graph generated by the Chung-Lu process with power law parameter $\beta$ and average expected degree $d$ (possibly depending on $n$, but satisfying $d = n^{o(1)}$). Let $u \neq v$ be a random pair of nodes in $G$. Any algorithm that receives node query access to $G$ starting at $u$ and $v$ must make, with high probability, $n^{\Omega(1)}$ queries to $G$ in order to find a path between $u$ and $v$. 
	\end{theorem}
	\begin{proofsketch}
		Consider a Chung-Lu graph with parameters $\beta,d$, and suppose that $n$ is large enough. The highest expected degree of a vertex in the graph is bounded by $n^{1 - c}$ for a constant $c$ that may depend on $\beta$ and $d$ (but not on $n$). 
		
		For a given time step of the algorithm, let $A_u$ denote the component consisting of all nodes queried by virtue of being reached from $u$ in the algorithm, and all edges in $G$ intersecting at least one such node. Define $A_v$ similarly for $v$. 
		Consider one step of the algorithm where some vertex $w$ is added to $A_u$, together with all new (previously unseen) edges incident to $w$. Let $e$ be any such new edge. The probability that $e$ intersects $A_v$ (and thereby closes the desired path between $u$ and $v$) is bounded by
		$$
		\sum_{x \in A_v}\frac{w(x)}{2W} \leq |A_v| \cdot \frac{n^{1-c}}{n} = \frac{|A_v|}{n^c},
		$$
		where $|A_v|$ is the number of \emph{edges} in the component $A_v$. Union bounding over all edges $e$ added in this process, we get that the probability of finding a path is bounded by $|A_u| \cdot |A_v| \cdot n^{-c}$. It follows that to find a path with good probability, at least one of $A_u$ and $A_v$ needs to have at least $n^{c/3}$ edges. The rest of the proof sketch is devoted to proving this last claim.
		
		We use the following fact on Chung-Lu graphs with the relevant parameters. Fix $\alpha > 0$. The probability that a new edge (emanating from a newly queried vertex $w$) will hit, at its other endpoint, a vertex of expected degree at least $n^{\alpha}$ is at most $n^{-\gamma}$, where $\gamma$ is a constant that depends only on  $\beta, d, \alpha$, and $\gamma\in(0,\alpha)$. 
		
		Pick $\alpha = c/6$ and the corresponding $\gamma < \alpha$, and suppose we make an arbitrary traversal involving $n^{\gamma/2}$ queries starting at $u$. The probability that any specific query hits a node of degree at least $n^\alpha$ is bounded by $n^{-\gamma}$. Union bounding over all $n^{\gamma/2}$ queries, we conclude that the event that an $n^{\alpha}$-degree node is queried throughout the whole traversal is bounded by $n^{\gamma/2} \cdot n^{-\gamma} = n^{-\gamma/2}$. Conditioning on this event not happening, the size of $A_u$ at the end of the traversal is bounded by $n^{\gamma / 2} \cdot n^{\alpha} < n^{c/3}$, as desired.

		The reasoning in the last paragraph implicitly assumes that outgoing edges from $A_u$ will never intersect $A_u$ itself. This is true with high probability by essentially the same argument, since the probability of each edge $e$ to intersect $A_u$ is bounded by $|A_u| / n^c < n^{-2c/3}$.
	\end{proofsketch}

	\section{Experimental Results}
	\begin{table}[b]
		\centering
		\begin{tabular}{|c||c|c|}
			\hline
			Class & $|\lz|$ & Networks \\
			\hline 
			small & 6\% & \smaller dblp, epinions, slashdot, skitter\\
			med & 4\% & \smaller large-dblp,  pokec, livejournal, orkut \\
			large & 1\% & \smaller wikipedia, soc-twitter\\
			\hline
		\end{tabular}
		\caption{Classification of networks used in experiments by size of $\lz$ used in experiments.}
\label{tab:classes}
\vspace{-.7cm}
\end{table}
\begin{table}[!htb]
\centering
\begin{tabular}{|c||c|c|c|c|c|c|}
	\hline
	Network & $|V|$ & $|E|$ & \small \bibfs{} & \small\PLL{} & \small\MLL{} \\
	\hline 
	epinions & $7.6\cdot 10^4$ & $5.1\cdot 10^5$ &\cmark & \cmark& \cmark\\
	slashdot & $7.9\cdot 10^4$ & $5.2\cdot10^5$ &\cmark & \cmark& \cmark\\
	dblp & $3.2\cdot 10^5$ & $1.0\cdot 10^6$ & \cmark & \cmark& \cmark\\
	skitter & $1.7\cdot 10^6$ & $1.1\cdot 10^7$ &\cmark & \cmark& \cmark\\ \hline
	large-dblp & $1.8\cdot10^6$ &$2.9\cdot10^7$ &\cmark & \cmark &\xmark\\
	soc-pokec & $1.6\cdot 10^6$ & $3.1\cdot 10^7$ &\cmark & \xmark &\xmark\\
	soc-live & $4.8\cdot 10^6$ & $6.8\cdot10^7$ &\cmark & \xmark&\xmark\\
	soc-orkut & $3.1\cdot10^6$ & $1.2\cdot10^8$ &\cmark & \xmark&\xmark\\
	\hline
	wikipedia & $1.4\cdot10^7$ & $4.4\cdot10^8$ &\cmark & \xmark&\xmark\\
	soc-twitter & $4.2\cdot 10^7$ & $1.5\cdot10^9$ &\cmark & \xmark&\xmark\\
	
	\hline
\end{tabular}
\caption[caption]{Network datasets used for experimental evaluation with their corresponding sizes. We observe that \bibfs{} finishes on all the datasets, but the indexing based methods do not on the medium and large networks. 
	
	\small{We were able to set up \MLL{} on large-dblp in reasonable time, but the subsequent shortest path inquiries were met with consistent segmentation faults that we were unable to debug.}} 
\label{tab:data}
\vspace{-.7cm}
\end{table}

In this section, we experimentally evaluate the performance of our algorithm. We look at several metrics to evaluate performance in different aspects. We compare with \bibfs{}, a traversal-based approach, and with the indexing algorithms \PLL{} and \MLL. We test several aspects, summarized next. Detailed results are provided in the rest of this section.
\begin{compactenum}
\item \textbf{Query cost: } By query cost, we refer to the number of vertices queried by \ouralg, consistent with our access model (see \Sec{setting}). We show that \ouralg{}  actually does remarkably well in terms of query cost, seeing a small fraction of the whole graph even for several thousands of shortest path inquiries.    See Figures~\ref{fig:summary2}(b) and~\ref{fig:qcost-small}.
\item \textbf{Inquiry time: } We demonstrate that \ouralg{}$_E$  achieves consistent speedups over traditional  \bibfs{}, even while using it as the sole primitive in the procedure. More complex methods such as \PLL{} and \MLL{} time out for the majority of large graphs. We also provide variants that achieve substantially higher speedups. Finally, in \Sec{prim}, we show how using the existing indexing-based state or the art methods on the core lets us achieve indexing-level inquiry times. See Figure~\ref{fig:exp-summary}.
\item \textbf{Accuracy: } We show that our estimated shortest paths are accurate up to an additive error 2 on 99\% of the inquiries for the default version  \ouralg{}$_E$; a faster heuristic, \ouralg{}$_H$, shows lower accuracy, but still over 90\% of inquiries satisfy this condition. See \Sec{vsbibfs} and ~\Tab{wormhole-summary} for details.
\item \textbf{Setup: } We look at the setup time and disk space with each associated method.  Perhaps as expected, \ouralg{}$_E$ beats the indexing based algorithms by a wide margin in terms of both space and time: see ~\Fig{setuptime}. In \Sec{prim} we further show that using these methods restricted to $\lz$ results in a variant \ouralg{}$_M$ with much lower setup cost (~\Tab{mlloncore}).
\end{compactenum}
\begin{table*}[!htb]
\centering
\begin{tabular}{|c||c||c|c|c|c|c||c|c|c|c|c|}
	\hline
	&\bibfs{} & \multicolumn{5}{c||}{\ouralg{}$_E$}& \multicolumn{5}{c|}{\ouralg{}$_H$} \\  \hline
	Network & MIT  & MIT & SU/I & +0(\%) & $\leq +1$  (\%) & $\leq +2$  (\%) & MIT & SU/I & +0(\%) & $\leq +1$  (\%) & $\leq +2$  (\%) \\
	\hline 
	epinions  & 144  &  41 & \cellcolor{blue!8} 4.5 &  \cellcolor{green!25}98.06   &  \cellcolor{green!25}99.99   &   \cellcolor{green!81}100.00  & 20 &\cellcolor{blue!15} 24 & 66.97   &  \cellcolor{green!25}99.54   &   \cellcolor{green!81}100.00    \\
	slashdot &  99 &  46 & \cellcolor{blue!8}2.8 &  73.43   &  \cellcolor{green!25}95.37   &  \cellcolor{green!25}99.28    &  24 &  \cellcolor{blue!8} 14  &63.09   &  \cellcolor{green!25}98.78   &  \cellcolor{green!25}99.98     \\
	dblp & 247  &  110 & \cellcolor{blue!8} 2.4 &  \cellcolor{green!25}97.02   &  \cellcolor{green!25}99.96   &   \cellcolor{green!81}100.00    & 48  &  \cellcolor{blue!8} 11  & 44.72   & 82.42   &  \cellcolor{green!25}96.53     \\
	skitter &  3004  & 1439 & \cellcolor{blue!8}2.3 &  \cellcolor{green!9}94.71   &  \cellcolor{green!25}99.89   &   \cellcolor{green!81}100.00    &660 & \cellcolor{blue!15} 24 & 58.99   &  \cellcolor{green!25}96.78   &  \cellcolor{green!25}99.98  \\ \hline
	large-dblp & 3041 &  1447 & \cellcolor{blue!8} 2.3 &  85.37   &  \cellcolor{green!25}99.10   &  \cellcolor{green!25}99.95    &417 &  \cellcolor{blue!15}21& 47.61   &  89.74   &  \cellcolor{green!25}99.04  \\
	pokec & 2142 & 1317 & \cellcolor{blue!8} 1.8 &  51.37   & \cellcolor{green!9} 92.15   &  \cellcolor{green!25}99.63    &506 &  \cellcolor{blue!8}11 & 14.52   &  59.51   &  \cellcolor{green!9}90.71  \\
	livejournal & 8565 &  4318 & \cellcolor{blue!8} 2.1 &  71.98   &  \cellcolor{green!25}97.95   &  \cellcolor{green!25}99.86    & 1054 & \cellcolor{blue!15} 29 & 28.86   &  77.93   &  \cellcolor{green!25}97.83  \\
	orkut &  14k & 3213 & \cellcolor{blue!8}4.4  &  58.50   & \cellcolor{green!9} 94.56   &  \cellcolor{green!25}99.64    & 1030  &  \cellcolor{blue!25} 35 & 20.66   &  68.11   & \cellcolor{green!9} 93.93  \\ \hline
	wikipedia & 35k & 17k & \cellcolor{blue!8} 2.4 & \cellcolor{green!9} 94.94   &  \cellcolor{green!25}99.92   & \cellcolor{green!81}  100.00   & 3394  &  \cellcolor{blue!25} 36 & 44.65   &  \cellcolor{green!25}98.74   &   \cellcolor{green!81}100.00   \\
	soc-twitter &  204k   & 81k & \cellcolor{blue!8}3.4 &\cellcolor{green!9}  93.30   & \cellcolor{green!25} 99.98   &  \cellcolor{green!81} 100.00   & 12k & \cellcolor{blue!35} 181 & 35.13 & \cellcolor{green!25} 99.30   & \cellcolor{green!25} 99.99    \\
	\hline
\end{tabular}

\caption{Summary of \ouralg{} with the two cases: \ouralg{}$_E$, with the exact shortest path through the inner ring, and \ouralg{}$_H$ that picks only the shortest path between the highest degree vertices -- refer to \Sec{vsbibfs}.  We note the mean inquiry times per inquiry (MIT) in microseconds, and average speed up \emph{per inquiry} (SU/I) compared to \bibfs{} for each method. We also note the percentiles of inquiries by absolute error: for \ouralg{}$_E$, we get absolute error under 2 for over 99\% of the inquiries. This drops for \ouralg{}$_H$, but it is still above 99\% for six of the ten datasets, and over 90\% in all of them. Accuracy numbers are highlighted in green, where darker is better. Similarly, we have a gradient of violet for speedups; darker is faster. For \ouralg{}$_E$, speedup over \bibfs{} per inquiry on average is usually between 2$\times$ and 3$\times$, but this increases to consistently between $20-30\times$ in \ouralg{}$_H$, and reaches a max of $181\times$ in our largest dataset, soc-twitter.
	\vspace{-.7cm}}
\label{tab:wormhole-summary}
\end{table*}
\paragraph{Datasets}

The experiments have been carried out on a series of datasets of varying sizes, as detailed in  \Tab{data}. The datasets have been taken either from the SNAP large networks database ~\cite{SNAP} or the KONECT project ~\cite{konect}. We organize the results into two broad sections: we first introduce two variants of our algorithm.  We then compare it with  \bibfs{} as well as  indexing based methods -- \PLL{} and \MLL{}. The latter two did not terminate in 12 hours for most of the graphs, while \bibfs{} completed on even our largest networks. 

We classify the examined graphs into three different classes and use a fixed percentage as the `optimal' inner ring size for graphs of comparable size (where the inner core size as $\%$ of the total size decreases for larger networks, an indication for the sublinearity of our approach). This takes into account the tradeoff between accuracy and the query/memory costs incurred by a larger inner ring. The classification is summarized in \Tab{classes}. For the experimental section, we default to these sizes unless mentioned otherwise. 

\paragraph{Implementation details}
We run our experiments on an AWS ec2 instance with 32 {\tt AMD EPYC\texttrademark{} 7R32} vCPUs and 64GB of RAM. The code is written in C++ and is available in the supplementary material as a zipped folder, with links to the datasets. The backbone of the graph algorithms is a subgraph counting library that uses compressed sparse representations~\cite{Escape}.

\subsection{\ouralg{}$_E$, \ouralg{}$_H$ and \bibfs} \label{sec:vsbibfs}We run two separate versions of \ouralg{}: the one, as described in \Alg{wormhole}, is what we refer to as \ouralg{}$_E$ (exact). Another variant we consider is where we pick just the highest degree vertices from $C(s)$ and $C(t)$ respectively, and do a BFS from those. This cuts down on the inquiry time, but also reduces the accuracy. We call this variant \ouralg{}$_H$. We take a deeper look into this tradeoff in the following sections. Note that both of these methods have the same query cost per shortest path inquiry, and the only difference is in running time (as all vertices in $\lz$ have already been queried during \coregen). For all runs, we do approximately 10,000 inquiries of uniformly chosen source and destination pairs (discarding disconnected pairs and other invalid inquiries). Our program runs on a compressed sparse representation (CSR) graph, and we store binary arrays for both $\lz$ and $\lo$ for practical purposes.

\paragraph{Query Cost}
To examine query cost, we look at the number of vertices seen by the algorithm over the first 5000 inquiries and compare it to  \bibfs. We direct the reader to \Fig{summary2}(b) for a summary of the query cost for our larger graphs, and \Fig{qcost-small} for the same in the smaller ones.  Consistently across all graphs,  \bibfs{} quickly views between 70\% and 100\% of the vertices in just a few hundred inquiries. In comparison, the query cost of \ouralg{}  is quite small for all networks: in the smaller networks, we see less than 30\% of the vertices even after 5000 inquiries, and in the larger ones this number is less than 10\% (in the largest ones, wikipedia and soc-twitter, it is <2\%).
\begin{figure}[h!]
\centering

{\includegraphics[width=0.45\textwidth]{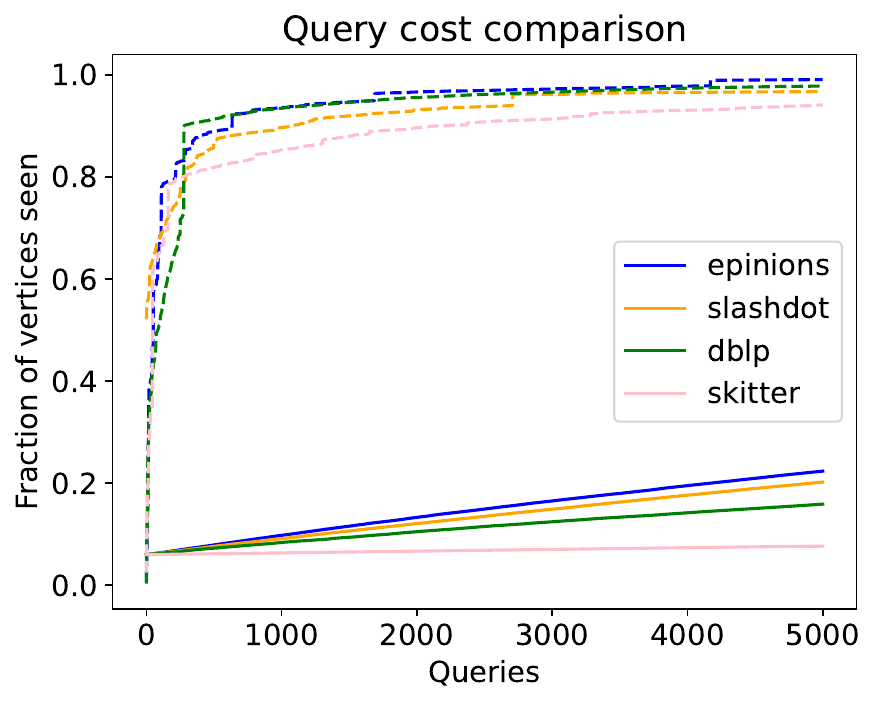}}

\vspace{-0.3cm}	
\caption{Query cost of \ouralg{}  and that of  \bibfs{} different datasets: fraction of the graph seen by \ouralg{}  vs  \bibfs{} over the first 10k inquiries in small and medium graphs. The dotted lines refer to the query cost by  \bibfs{} while the solid lines are due to \ouralg{}. The results for large and huge graphs appear in item (b) of Figure~\ref{fig:exp-summary}.
	\vspace{-0.3cm}
} \label{fig:qcost-small}
\end{figure}

\paragraph{Accuracy}
\label{sec:accuracy}
We consider two error measures, absolute (additive) and relative (multiplicative).
Clearly, an additive error of, say, 2, is less preferable for shorter paths than for longer ones. The \emph{relative error} $re(s,t)$ is more refined, as it measures the error with respect to the actual distance of the pair at question. Formally, 
\begin{align}
re(s,t) = \dfrac{\overline d(s,t) - d(s,t)}{d(s,t)},
\end{align}
where for a vertex pair $(s,t)$, $d(s,t)$ is the true distance and $\overline d(s,t)$ is the approximate distance estimated by \ouralg.
We investigate the relative error as a function of the core size -- see \Fig{errors}. 
In general, the accuracy drops as we decrease the size of the core. Moreover, we observe that larger graphs give comparably good results at much smaller inner ring sizes, keeping in line with our hypothesis of a sublinear inner ring.

The other key accuracy statistic is additive error: this is summarized in \Tab{wormhole-summary}. For $\ouralg{}_E$ across almost all networks, the vast majority of the pair inquiries are  \emph{estimated perfectly}. Our worst performance is on soc-pokec and soc-live. Even there, we have perfect estimates for 60\% of vertices, and over 94\% of vertices have an additive error of less than 1. In \emph{all networks}, more than 99\% of the pairs are estimated with absolute error lesser or equal to 2 in \ouralg{}$_E$. The accuracy is poorer in \ouralg{}$_H$ (recall that in this variant we do not compute all-pairs shortest paths in the core, resorting to an approximate heuristic instead). However, we note that even then, in most graphs, we have an additive error of at most 2 in over 99\% of the queries, and over 90\% for all graphs. 

\paragraph{Speedups over \bibfs{} in inquiry time}
The main utility of \ouralg{}$_H$ is in exhibiting how much faster our algorithm becomes if we sacrifice some  accuracy. This is also documented in \Tab{wormhole-summary}. \ouralg{}$_E$ already achieves speedups per inquiry over \bibfs{}: typically at least 2$\times$, but up to over 4$\times$ in some networks. The variant  \ouralg{}$_H$ further  speeds up each inquiry by another order of magnitude, up to a massive 181$\times$ in case of our largest network, soc-twitter (while utilizing the same decomposition and data structure). We  note that \bibfs{} does not have any setup cost, while \ouralg{} does (for both variants). However, we show that our setup costs are typically very low: our highest setup time is about two minutes, and the highest setup space requirement is under 100MB. The complete statistics are provided in \Tab{setupworm}. (In our implementation, setup time is the time needed for us to capture $\lz$ and $\lo$ in Algorithm \coregen, and setup space is the stored binary files for $\lz$ and $\lo$.)

\begin{figure}[]
\centering

\begin{subfigure}[Mean relative error with varying inner ring sizes across different networks]
	{\includegraphics[width=0.4\textwidth]{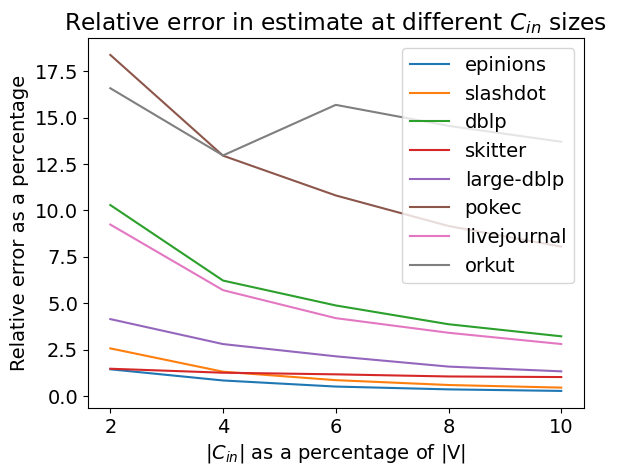}}
\end{subfigure}
\begin{subfigure}[Mean relative error with varying inner ring sizes for soc-twitter]
	{\includegraphics[width=0.4\textwidth]{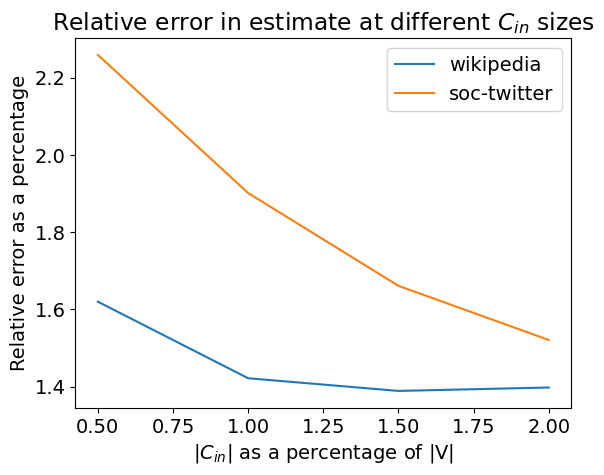}}
\end{subfigure}

\caption{Accuracy of \ouralg{}$_E$  in different settings across different datasets. We plot relative error  at different inner ring sizes for the datasets.} \label{fig:errors}
\vspace{-.3cm}
\end{figure}

\begin{table*}[]
\begin{tabular}{|c||c|c|c|c|c|c|c|c|c|c|}
	\hline
	Metric & \smaller epinions & \smaller slashdot &\smaller dblp &\smaller skitter & \smaller large-dblp & \smaller pokec & \smaller livejournal & \smaller orkut & \smaller wikipedia & \smaller soc-twitter  \\ \hline
	Setup time (s) & 0.02 & 0.02 & 0.03 & 0.42 & 0.40 & 0.46 & 1.68 & 1.49 & 11.80 & 123.57 \\
	$\lz+\lo$(MB) &0.15 & 0.15 &0.62 & 3.31 & 3.56 & 3.19 & 9.47 & 6.00 & 26.56 & 81.35 \\
	\hline
\end{tabular}
\caption{Setup cost for \ouralg{}: this holds for both \ouralg{}$_E$ and \ouralg{}$_H$. Setup time is the time needed  to capture $\lz$ and $\lo$ in \Alg{core}. The last row, space, is the footprint on disk of our binary arrays.
	\vspace{-0.5cm}
}\label{tab:setupworm}
\end{table*}

\begin{table}[]
\centering
\begin{tabular}{|c||c|c||c|c||c|c|}
	\hline
	\multirow{2}{*}{Network} & \multicolumn{2}{c||}{Setup (sec)} & \multicolumn{2}{c||}{Inq. time ($\mu s$)} & \multicolumn{2}{c|}{Breakeven}\\ \cline{2-7}
	& \PLL & \MLL& \PLL & \MLL& \PLL & \MLL  \\ \hline
	epinions &  4.1 &1.5 & 0.96 &2.66 & 101k & 39k  \\
	slashdot &  6.8 &3.6 & 1.08 & 3.98 & 151k & 85k \\
	dblp & 218 
	&52.4 & 4.05 & 11.99 & 2.1M & 535k \\
	skitter &  
	1.1k 
	& 466 
	& 2.72& 11.06& 769k & 326k \\
	large-dblp &  9.2k 
	& 1.6k 
	& 9.25 & N/A & 6.4M& inf\\ \hline

	\hline
\end{tabular}
\caption{Comparisons with \PLL{} and \MLL{}. We look at setup time, mean inquiry time, and breakeven compared to \ouralg{}$_E$. The indexing based methods do not terminate on graphs larger than this. For large-dblp, setup completes for \MLL{} but we are unable to make inquiries.
	\vspace{-0.5cm}}
\label{tab:LLsetups}
\end{table}

\subsection{Comparison with index-based methods}\label{sec:vsLL}
As discussed, index construction allows for much faster inquiry times, but setup times that can take up to several hours even for relatively small-sized graphs. We attempt to benchmark our algorithm against two state of the art methods, \PLL{} and \MLL. \PLL{} solves the easier task of finding distances, while \MLL{} does explicit shortest path construction (the authors note that \PLL{} may be extended to output paths, but no code is publicly available). However, once the graphs hit a few million vertices, these methods either take too long to run the setup, or even if they succeed in index construction, they may be too large to load into memory. We summarize the results in \Tab{LLsetups}, and expand on the discussion in the following paragraphs.

\paragraph{Mean inquiry time} In the cases where index based methods do succeed (limited to graphs with fewer than 30 million edges), they have a clear advantage in per inquiry cost.  Their typical inquiry time is in the microsecond range, where \PLL{} is faster since it only computes distances, and \MLL{} is about 3 times slower than \PLL{}.

\paragraph{Setup cost} In terms of setup times, both \ouralg{}$_E$ and \ouralg$_H$ have a massive advantage over the index-based methods. We direct the reader to \Tab{setupworm} and \Tab{LLsetups} for a comparison of setup times: we let all methods run for 12 hours, and terminate if they do not finish in that time. Both \PLL{} and \MLL{} failed to complete setup  for any graph with more than 30 million vertices. In comparison, even for our largest graph, soc-twitter, of over a 1.5 billion edges, the setup time for \ouralg{}$_E$ is just minutes. Even in the cases where these methods do terminate, the storage footprint  is massive. We observed that if allowed to run, \MLL{} completes index construction on soc-pokec in a little under 24 hours, but the constructed files are almost a combined 45 gigabytes in size; in comparison, the input COO file is only 250 megabytes! A detailed comparison of the space footprint is given in  \Fig{summary2} in \Sec{intro}, and the time comparisons can be found in \Fig{setuptime}, and in \Tab{LLsetups}.
\begin{figure}[h!]
\centering
\includegraphics[width = 0.4\textwidth]{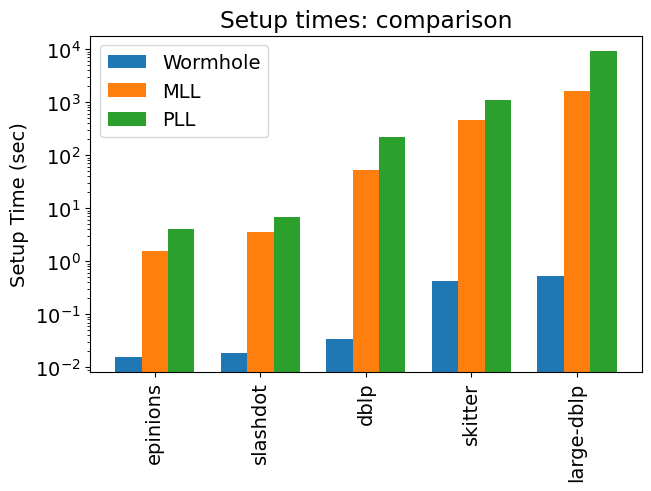}
\caption{A comparison of setup time between different methods. The index-based methods did not terminate on graphs larger than these.
	\vspace{-0.5cm}
} \label{fig:setuptime}
\end{figure}

\paragraph{When is indexing better? A running time comparison:}
Given the very high setup time of index-based methods, \ouralg{}$_E$ has a head start. However, index-based solutions do catch up and outdo \ouralg{} after sufficiently many shortest path inquiries. We quantify this threshold to give the reader a sense of when to favor each approach. 
We refer to the threshold as 
\emph{Breakeven} (\texttt{BE}$_X$ for a competing method $X$) and provide  its values for different graphs in \Tab{LLsetups}. Since the index-based methods have setup times between $10^6$ and $10^9$ orders of magnitude higher than the inquiry times, we look at how many inquiries are needed for them to have an average gain over \ouralg{}$_E$ in the net time taken. Formally, we define breakeven for \PLL (likewise \MLL) as:

$$\texttt{BE}_\PLL = \frac{\textrm{setup}_{\PLL} - \textrm{setup}_{\ouralg_{E}}}{MIT_{\ouralg_{E}}-MIT_\PLL}.\label{eq:BE}$$
We can similarly compute the breakeven with respect to \ouralg$_H$, but we note that it is already quite high even against the much slower version \ouralg$_E$. This implies that even with the low time per inquiry, the setup cost is so prohibitively high that the gains take hundreds of thousands to even millions of inquiries to set in, even on the small networks.

\subsection{\ouralg{}  as a primitive: \ouralg{}$_M$}\label{sec:prim}
\ouralg{}$_E$ and $\ouralg{}_H$ perform well in terms of query cost and accuracy. The inquiry times, especially in the latter variant, are also huge improvements over  \bibfs{}, but lag behind indexing based methods that perform lookups to find shortest paths. However, as evident by our experiments, landmark based index creation is often prohibitively expensive, both in terms of the time taken to create the index and the space required to store it, to the extent that it may even be impossible for large networks. The success of \ouralg{}$_E$ comes from exact shortest paths computed solely on a small core, which is as low as 1\% of the graph. In practice, it takes very little time for a traversal based algorithm to reach $\lz$, and the bulk of the cost comes from the exact path computation inside $\lz$. We thus ask, how much faster can we make our algorithm if we  speed this process up, perhaps by doing indexing solely on the core?

\begin{table}[]
\begin{tabular}{|c||c|c|c|}
	\hline
	\multirow{2}{*}{ Network} & \multicolumn{2}{c|}{Setup for \MLL{} on $\lz$}& \multirow{2}{*}{MIT ($\mu s$)} \\ \cline{2-3}
	& Time (sec)& Space (MB) & \\ \hline
	epinions & 0.41 & 2 & 1.57\\
	slashdot & 0.54 & 5 & 2.45\\
	dblp & 2.99 & 20 & 5.00\\
	skitter & 28.24 & 106 & 12.04\\
	large-dblp & 55.68 & 182 & 15.22\\
	pokec & 144.34 & 328 & 45.52\\
	livejournal & 803.95 & 1303 & 57.89\\
	orkut & 1156.28 & 1476 & 157.67\\
	wikipedia & 551.19 & 452 & 120.65\\
	soc-twitter & 20949.82 & 6215 & 115.44\\
	\hline
\end{tabular}
\caption{\ouralg{}$_M$: running \MLL{} on $\lz$. 
}\label{tab:mlloncore}
\end{table}

To this end, we propose as a third alternative, \ouralg{}$_M$: it functions almost identical to \ouralg{}$_H$, but instead of running \bibfs{} to find the shortest paths in $\lz$, it sets up \MLL{} on all of $\lz$ and then uses the \MLL{} index to compute shortest paths  inside $\lz$. This is an illustrative example to show how our decomposition can be used in combination with existing techniques. The accuracy guarantees of this variant as presented will be identical to \ouralg{}$_H$; however, we do not analyze the index size theoretically and suspect it will not be sublinear. 

We conduct similar experiments for \ouralg{}$_M$. Remarkably,  while \MLL{} fails to complete setup on most of the graphs, \ouralg$_M$ successfully  runs it on the core in \emph{all} cases. Moreover, as noted in \Tab{mlloncore}, the cost in both time and space is orders of magnitude smaller than for the full graph, though still significantly larger than the default \ouralg{}$_H$. We note about two orders of magnitude of improvement in time per inquiry over \ouralg{}$_H$, but at the same time, the setup cost is also about two orders of magnitude higher in both time and space. Notably, even in cases where \MLL{} does complete on the full graph, we are able to answer inquiries in roughly the same time (see \Fig{exp-summary}) at a fraction of the setup cost (\Fig{setuptime}). We leave a more systematic investigation of this approach of combining \ouralg{} with  existing methods on the core to future work.

\bibliographystyle{ACM-Reference-Format}
\bibliography{literature}

\end{document}